\documentclass[11pt]{article}

\usepackage[english]{babel}
\usepackage[utf8]{inputenc}
\usepackage[authoryear,longnamesfirst]{natbib}
\usepackage[margin=1in]{geometry}
\usepackage{setspace}
\usepackage{subcaption}
\usepackage{microtype}

\usepackage{amsmath}
\usepackage{amsfonts}
\usepackage{amsthm}
\usepackage{amssymb}
\usepackage{mathrsfs}
\allowdisplaybreaks
\usepackage{amstext}
\usepackage{comment}
\usepackage{epsfig}
\usepackage{latexsym}
\usepackage{epstopdf}
\usepackage{bm}
\usepackage{algorithm}
\usepackage{algpseudocode}

\usepackage{footnote}

\usepackage{accents}
\usepackage{url}
\usepackage[colorlinks = true,linkcolor = blue]{hyperref}
\hypersetup{
        urlcolor = blue,
	colorlinks = true,
	linkcolor = blue,
	citecolor = blue,
}
\usepackage{enumitem}

\usepackage{multirow}
\usepackage{multicol}
\usepackage{babel}
\usepackage{array}
\usepackage{rotating}
\usepackage{longtable}
\usepackage{float}
\usepackage{booktabs}
\usepackage{lscape}

\newtheorem{theorem}{Theorem}

\newtheorem{assumption}{Assumption}
\newtheorem{lemma}{Lemma}

\newtheorem{example}{Example}

\newtheorem{corollary}{Corollary}

\newcommand{\pii}{\varsigma}

\DeclareMathOperator*{\argmax}{argmax}
\DeclareMathOperator*{\argmin}{argmin}

\newcommand{\E}{\mathbb{E}}

\title{Data-Driven Policy Learning for Continuous Treatments\thanks{Alphabetical ordering of authors: all authors contributed equally to this work.}}

\author{Chunrong Ai\thanks{School of Management and Economics, The Chinese University of Hong Kong, Shenzhen. Email: \url{chunrongai@cuhk.edu.cn}.} \quad Yue Fang\thanks{School of Management and Economics, The Chinese University of Hong Kong, Shenzhen. Email: \url{fangyue@cuhk.edu.cn}.} \quad Haitian Xie\thanks{Guanghua School of Management, Peking University. Email: \url{xht@gsm.pku.edu.cn}.}}

\date{\today}

\begin{document}

\maketitle

\begin{abstract}
    \onehalfspacing
   
     This paper studies policy learning for continuous treatments from observational data. Continuous treatments present more significant challenges than discrete ones because population welfare may need nonparametric estimation, and policy space may be infinite-dimensional and may satisfy shape restrictions. We propose to approximate the policy space with a sequence of finite-dimensional spaces and, for any given policy, obtain the empirical welfare by applying the kernel method. We consider two cases: known and unknown propensity scores. In the latter case, we allow for machine learning of the propensity score and modify the empirical welfare to account for the effect of machine learning. The learned policy maximizes the empirical welfare or the modified empirical welfare over the approximating space. In both cases, we modify the penalty algorithm proposed in  \cite{mbakop2021model} to data-automate the tuning parameters (i.e., bandwidth and dimension of the approximating space) and establish an oracle inequality for the welfare regret. 

    
\end{abstract}

\newpage

\section{Introduction}

Economists are increasingly interested in learning optimal policy from experimental and observational data. The optimal policy maximizes the population welfare over a (possibly restricted) policy space, where a policy maps individual characteristics into a policy treatment. Computing the optimal policy may encounter two challenges. The population welfare may be unknown to the policymaker, and the other is that the policy space may be infinite-dimensional and complex.  One general approach is to apply various methods, including the general methodology proposed by \cite{ainon}, to estimate the population welfare from observational data to obtain empirical welfare and approximate the complex and infinite-dimensional policy space with a sequence of finite spaces \citep[e.g.,][]{ai2003efficient}. The learned policy then maximizes the empirical welfare over the approximating space.  This general approach may introduce tuning parameters in estimating the population welfare (e.g., bandwidth in kernel estimation of the population welfare) and the approximation (e.g., the dimension of the approximating space).  The learned policy requires careful calibration of the tuning parameters to achieve an oracle inequality of welfare regret.  

The existing literature on policy learning from observational data has followed the general approach above, but focused mainly on binary treatments. The binary treatment setting has two advantages over general settings. First, the empirical welfare is a simple sample average that is unbiased and free of tuning parameters. Second, the optimal policy may have an analytical expression, thereby does not need approximation. For example, when the policy space is unrestricted, the optimal policy is an indicator function of the conditional average treatment effect (hereafter, CATE), which can be machine-learned from observational data \citep{manski2004statistical,manski2007admissible,manski2007minimax,stoye2009minimax,stoye2012minimax,tetenov2012statistical,bhattacharya2012inferring}. Under the condition that the machine-learned CATE converges to the truth fast, these studies established a sharp upper bound of the welfare regret. However, when the policy space is restricted, the optimal policy generally does not have an analytical expression, even in the binary setting.  Instead of approximating the policy space, \cite{kitagawa2018should,kitagawa2021equality} parameterized it as a finite, fixed-dimensional space. \cite{athey2021policy,zhou2023offline} also parameterized it, but allowed the dimension to grow with the sample size. None of those studies has any tuning parameters. Under sufficient conditions, they all established a sharp (i.e., minimax-optimal rate) bound of the welfare regret. \cite{mbakop2021model}, on the other hand, did not parameterize the policy space and used approximations. The approximation introduces one tuning parameter: the dimension of the approximating space. They suggested a penalized algorithm to data-automate the dimension. However, due to the approximation error, they could not achieve the same rate as in \cite{kitagawa2018should,athey2021policy}, obtaining instead an oracle inequality of the welfare regret that balances approximation and estimation errors.\footnote{\citet{fang2025model} used doubly robust moment conditions for welfare estimation and obtained similar results for multivalued discrete treatments.}

In a real-world context, policymakers often use complex policy treatments.  For example, they set carbon tax rates or allocate pollution permits in environmental policy, determine the duration of training for various demographic groups in job training programs, adjust cash transfer amounts across different households in conditional cash transfer programs, and set different price levels for different customer groups in retailing. All these policy treatments are continuous.  Yet, learning the optimal continuous policy from observational data has received scant attention in the literature.\footnote{Several recent studies have examined continuous treatments in various causal frameworks, such as \citet{su2019non,callaway2021difference,xie2024nonlinear,ColangeloLee2025}, with a focus on estimating treatment effects or dose–response functions, which differs from the policy learning objective considered here.} The difficulty is in estimating the population welfare because there are few observations at each level of treatment. In principle, policymakers can use observations in the neighborhood of each treatment level to evaluate the policy. But doing so introduces the bias and another tuning parameter (e.g., the neighborhood's size). \cite{kallus2018policy} took this approach with kernel estimation but considered a finite, fixed-dimensional policy space. So, they have only one tuning parameter, the bandwidth. They did not data-automate the bandwidth and established only an upper bound of the welfare regret, not the oracle inequality.   

We are unaware of any work on learning the optimal continuous policy from observational data, with policy space approximated. This paper intends to fill the literature gap. Specifically, we generalize \cite{mbakop2021model} to a continuous treatment setting by approximating the policy space with finite, growing spaces and applying the kernel method to obtain the empirical welfare. We consider two cases: known and unknown propensity scores. We allow machine learning propensity scores and modify the empirical welfare to account for the machine learning effect. We then maximize the empirical welfare or modified empirical welfare over the approximating space to obtain the learned policy, which depends on the tuning parameters. We then modify the penalized procedure of \cite{mbakop2021model} and develop a data-automated algorithm for both tuning parameters. Despite the extra tuning parameter, we still establish an oracle inequality of the welfare regret in known and unknown propensity scores. Our oracle inequality in the known propensity score case is similar to that of \cite{mbakop2021model}. Still, it is sharper in the unknown propensity score case because we use the double debiasing approach.

The extension is nontrivial because the extra tuning parameter (i.e., bandwidth) complicates the data automation algorithm. The bandwidth and the approximating space dimension play different roles. While the choice of dimension directly impacts policy learning performance, bandwidth directly affects policy evaluation performance. The data automation algorithm must consider the individual effects to maximize the policy evaluation and learning performance separately and the interaction effects of the tuning parameters on the learned policy. Despite the complicated calibration procedure, it is worthwhile to automate the tuning parameters since they adapt to the underlying data-generating process without knowing the model's smoothness condition.

To illustrate the practical value of the proposed policy learning, we re-examine the policy of assigning individuals to job training programs of varying durations. We use the same data from the Job Training Partnership Act (JTPA) study as \cite{kitagawa2018should,mbakop2021model}. While they analyze the binary treatment (i.e., participation in job training), we examine the continuous treatment (i.e., training duration). \cite{flores2012estimating} noted that the effects of job training on future earnings may vary with the length of exposure to the training program. It is crucial to consider different training durations (as opposed to a binary participation decision) in policy design. Our findings confirm that the learned policy adapts to and reflects the varying training time duration, further highlighting the benefits of moving beyond binary participation frameworks to develop data-driven approaches to policy design.

We organize the remainder of the paper as follows. Section \ref{sec:welfare} sets up the model. Section \ref{sec:ipw} introduces the data-automation algorithm in the known propensity score setting and establishes the oracle inequalities. Section \ref{sec:dr} extends the analysis to the unknown propensity score setting. Section \ref{sec:sieve} discusses examples of policy space approximations. Section \ref{sec:applications} presents an empirical study. The proofs for theoretical results in the main text are collected in the Appendices.

\section{Setup and Notation} \label{sec:welfare}

\subsection{Population model}\label{sec:population_model}

The model consists of a continuous treatment $T$ with support $\mathcal{T} \subset \mathbb{R}$, a set of potential outcomes $\{Y(t)\}_{t\in\mathcal{T}}$, and a vector of covariates $X$ with support $\mathcal{X} \subset \mathbb{R}^{d_X}$. The researcher only observes the realized outcome $Y\equiv Y(T)$, not all potential outcomes. A policy $\pi$ maps the covariate space $\mathcal{X}$ to the treatment space $\mathcal{T}$.

We define the population welfare of a policy $\pi$ as the expected outcome under this policy:
\begin{align} \label{eqn:def-welfare}
      W(\pi) \equiv \E[Y(\pi(X))].
\end{align}
We aim to find the optimal policy within an infinite-dimensional (and possibly restricted) space $\Pi_\infty$. Let $W^*(\Pi_\infty) \equiv \sup_{\pi \in \Pi_\infty} W(\pi)$ denote the global optimal welfare. Following the literature \citep[e.g.,][]{manski2004statistical,kitagawa2018should,athey2021policy}, we assess the performance of a policy $\pi$ by the welfare regret, $W^*(\Pi_\infty) - W(\pi)$, the difference between the global optimal welfare and the welfare achieved by $\pi$.  

As \cite{mbakop2021model} explains, economic theory or intuition often imposes nonparametric restrictions on policy classes, such as monotonicity, convexity, super-modularity, or separability. These restrictions restrict the form of the policies but still allow the policy class to remain infinite-dimensional. Another type of restriction is a parameterization of policy. Although parametric restriction reduces the policy space to finite-dimensional, they are generally ad hoc and seldom driven by economic theory. Below, we present some economically meaningful examples. 

\begin{example}\label{eg:single}
Consider the set of policies with a single-index representation:
\begin{align*}
   \Pi_\infty = \{ \pi(x) = h(x'\beta): h \text{ continuously differentiable}, \beta \in \mathbb{R}^{d_X} \}.
\end{align*}
The treatment assignment is based on a score (i.e., a linear transformation of individual covariates). Such single-score treatment rules are widely employed in firms' marketing strategies, as noted in \cite{hartmann2011identifying}.
\end{example}

\begin{example} \label{eg:separable}
    Consider the following separable and monotone policy class
    \begin{align*}
        \Pi_\infty = \left\{ \pi(x) = \textstyle \sum_{p=1}^{d_X} h_p(x_p): h_p \text{ decreasing, } 1 \leq p \leq d_X \right\},
    \end{align*}
    where $x_p$ denotes the $p$th coordinate of the covariates vector. The monotonicity of the function \( h_p \) may arise from fairness concerns, specifically that individuals endowed with higher values of \( X \) should not receive a higher level of treatment than those with lower \( X \). The separability structure ensures that the decrease in treatment assignment resulting from an increase in one covariate \( X_p \) does not depend on other covariates. This policy class can be regarded as the continuous-treatment analog used in the empirical study of \cite{mbakop2021model} and is implemented in our empirical study.
\end{example}

Researchers often consider continuous piecewise-linear policies defined by a set of thresholds. These rules are simple to communicate and implement. Block-rate taxes in environmental economics \citep{zhou2019would} and California’s two-tier price collar for carbon emissions are prominent examples, while analogous ``step-up'' structures appear in loyalty programmes: multi-threshold linear bonuses \citep{fang2018loyalty} and tiered status schemes in retail \citep{nishio2022joint} reward early engagement generously before tapering once customers are ``locked in''. Implementing such a policy requires choosing the number and placement of thresholds and the slope within each segment. We formalize this policy class in the next example.

\begin{example}\label{eg:piecewise_linear}
Consider the following piecewise-linear policy class.
Let $x\in\mathbb{R}$ be a one-dimensional score, possibly obtained by transforming a multivariate covariate vector.
For an integer $k\ge 0$, choose $k$ thresholds (location of kinks): $-\infty=s_{0}<s_{1}<\dots<s_{k}<s_{k+1}=\infty,$ and segment-specific intercepts $\alpha_{j}$ and slopes $\beta_{j}$ for
$j=0,\dots ,k$. The resulting policy $\pi$ is
$$
\pi(x)=\sum_{j=0}^{k}(\alpha_{j}+\beta_{j}x)\,\mathbf 1\{s_{j}\le x < s_{j+1}\},
$$
subject to the continuity constraints $\alpha_{j}+\beta_{j}s_{j+1}=\alpha_{j+1}+\beta_{j+1}s_{j+1},j=0,\dots ,k-1.$
Let $\Pi_{k}$ be the collection of all such policies with exactly $k+1$ segments.
Given an upper bound $K\in\mathbb{N}_{+}\cup\{\infty\}$, define the global policy class $\Pi_{K}=\bigcup_{k=0}^{K}\Pi_{k}.$
When $K=\infty$, the class allows an unrestricted number of segments; when $K<\infty$, the number of segments is capped at $K+1$ to reflect practical limits on simplicity and administrative capacity.

\end{example}
%For ease of illustration, we illustrate the policy with continuous covariates. However, no assumption in the identification, estimation, or calculation of VC dimension rules out discrete variables. All results in this paper apply to a combination of continuous and discrete covariates. If binary variables are included, then finite values of the covariate need to be considered to pick up the level of treatment. If all the covariates used in the policy are discrete, then $\Pi_\infty$ is a finite collection of policies.  %For example, considering the monotone policies, when a binary variable is included, the monotone policy simply evaluates two discrete points of that coordinate. 

Our theory accommodates all types of covariates—continuous, discrete, or mixed. Discrete covariates are theoretically more straightforward, and when all covariates are discrete, the set of all measurable policies is a finite-dimensional space.

\subsection{Empirical welfare } \label{sec:empirical-welfare}

We observe an independent and identically distributed (iid) sample $S_n \equiv \{(Y_i,T_i,X_i): 1 \leq i \leq n\}$ drawn from the distribution of the random variables $(Y,T,X)$. 

To estimate $W$, let us first recall that, with a discrete treatment, the IPW formula can determine the welfare (under the unconfoundedness condition; see Assumption \ref{ass:unconfoundedness} below) as\footnote{Because the following formula applies only when the treatment is discrete, we do not denote it by $W$.}
\begin{align*} 
    \mathbb{E}\left[ \mathbf{1}\{T = \pi(X)\} \frac{Y}{f(T|X)} \right], 
\end{align*}
where $f(t|x)$ denotes the generalized propensity score, i.e., the conditional density of the treatment given the covariates.
%In this case, the propensity score $f$ is a probability mass function. 
However, this formula is no longer valid in the continuous setting because the indicator $\mathbf{1}\{T=\pi(X)\}$ equals zero almost surely when $T$ is a continuous variable. We propose a kernel weighting, 
\begin{align*} %\label{eqn:continuous-IPW}
    W_h(\pi) = \mathbb{E}\left[ \frac{1}{h} K\left( \frac{T-\pi(X)}{h} \right) \frac{Y}{f(T|X)} \right],
\end{align*}
with $K$ as a kernel function and $h$ as the bandwidth. This leads to the following empirical welfare estimator:
\begin{align*} 
    \hat{W}_h(\pi) = \frac{1}{nh}\sum_{i=1}^n  K\left( \frac{T_i-\pi(X_i)}{h} \right) \frac{Y_i}{f(T_i|X_i)}.
\end{align*}

As directly optimizing the infinite-dimensional policy class $\Pi_\infty$ is not practically feasible, we shall use the sieves approximation to the global policy class \( \Pi_\infty \). This is achieved using a nested sequence of low-complexity policy classes $\{\Pi_k : k \geq 1\}$, where \( \Pi_k \subset \Pi_{k+1} \subset \cdots \subset \Pi_\infty \). We call each \( \Pi_k \) a sieve policy class to distinguish it from the global class \( \Pi_\infty \). Each $\Pi_k$ has finite complexity, measured by the Vapnik–Chervonenkis (VC) dimension\footnote{Let $\mathscr{S}$ be a collection of subsets of a set $\Upsilon$. The VC dimension of $\mathscr{S}$ is defined as the largest cardinality of a subset $\upsilon \subset \Upsilon$ that can be shattered by $\mathscr{C}$. The collection $\mathscr{S}$ is said to shatter $\upsilon$ if for each $\tilde{\upsilon} \subset \upsilon$, there exists a set $\mathcal{S} \in \mathscr{S}$ such that $\tilde{\upsilon} = \upsilon \cap \mathcal{S}$. Notice that this definition of VC dimension is the same as in \cite{van2009note,kitagawa2018should,wainwright_2019,mbakop2021model}, but is smaller by one than the VC dimension defined in \cite{wellner1996, dudley_1999}.} of the class of subgraphs $\{ \{(x,t):t<\pi(x) \} : \pi\in\Pi \}$. Throughout the paper, we use $\operatorname{VC}(\Pi)$ to denote the VC-subgraph dimension for a generic policy class $\Pi$.
% \textcolor{red}{explain shatter, also check if we are using VC or VC subgraph} 

Each sieve optimal policy estimator $\hat{\pi}_{h,k}$ is obtained by maximizing the estimated welfare $\hat{W}_h$ within each sieve policy class $\Pi_k$ as defined later in (\ref{eqn:sieve-EWM}).
We want to data-automate $(h,k)$ in a way that delivers a policy estimator $\hat{\pi}_{\hat{h},\hat{k}}$ with good performance. Below, we provide a heuristic discussion of the challenge.

The roles of $h$ and $k$ are asymmetric: $h$ is involved in the evaluation (welfare estimation) stage, whereas $k$ is used during the policy design stage.
Examine the following decomposition of the welfare regret:
\begin{align} \label{eqn:welfare-decomposition}
    W^*(\Pi_\infty) - W(\hat{\pi}_{h,k}) & = W^*(\Pi_\infty) - W^*(\Pi_k) + W^*(\Pi_k) - W(\hat{\pi}_{h,k}) \nonumber \\
    & \leq W^*(\Pi_\infty) - W^*(\Pi_k) + 2 \sup_{\pi \in \Pi_k}|\hat{W}_h(\pi) - W(\pi)| \nonumber \\
    & \leq \underbrace{W^*(\Pi_\infty) - W^*(\Pi_k)}_{\text{welfare deficiency}} + \underbrace{2 \sup_{\pi \in \Pi_k}|\hat{W}_h(\pi) - W_h(\pi)|}_{\text{variance}} + \underbrace{2 \sup_{\pi \in \Pi_k}|W_h(\pi) - W(\pi)|}_{\text{kernel bias}},
\end{align}
where $W^*(\Pi_k) \equiv \sup_{\pi \in \Pi_k} W(\pi)$ denotes the optimal welfare over $\Pi_k$, and the second inequality is (2.2) in \cite{kitagawa2018should}. 

We provide a heuristic argument for how the two tuning parameters affect the three terms on the right-hand side of (\ref{eqn:welfare-decomposition}). The welfare-deficiency term captures the loss (compared to global optimal welfare) from restricting to the sieve class $\Pi_k$; it depends only on $k$ and decreases as $k$ increases. The variance term, by empirical-process results in \cite{kitagawa2018should}, is of order $\sqrt{\operatorname{VC}(\Pi_k)/(n h)}$. The kernel bias, by standard analysis, depends only on the bandwidth and is of order $h^r$, where $r$ denotes the smoothness of the (conditional) dose-response functions defined as 
 \begin{align} \label{eqn:m-def}
        m(t,x) \equiv \mathbb{E}[Y(t)|X=x],
\end{align}
which describes the (conditional) mean potential outcome under each treatment level. 

To summarize, the right-hand side of (\ref{eqn:welfare-decomposition}) is proportional to the following order-of-magnitude:
\begin{align*}
    \text{welfare deficiency}(k) + \sqrt{\frac{\operatorname{VC}(\Pi_k)}{nh}} + h^r.
\end{align*}
In terms of minimizing the above sum of three terms, the optimal $(h,k)$ should be chosen such that the three terms are of the same order:
\begin{align*}
    \text{welfare deficiency}(k) \asymp \sqrt{\frac{\operatorname{VC}(\Pi_k)}{nh}} \asymp h^r,
\end{align*}
which yields $h \asymp \left( \frac{\operatorname{VC}(\Pi_k)}{n}\right)^{\frac{1}{2r+1}}$. Hence, the optimal choice of $h$ is dependent on $k$, and vice versa. Intuitively, enlarging the policy class raises the variance term $\sqrt{\frac{\operatorname{VC}(\Pi_k)}{nh}}$,\footnote{The phenomenon that a more complex policy class leads to higher variance is the so-called overfitting issue. When $\Pi_k$ is too rich relative to the sample size, the learned policy can fit sampling noise rather than genuine welfare signals, inflating the estimation error.} requiring a larger bandwidth to mitigate this increase.

This interplay illustrates a central difficulty in continuous-treatment policy learning: the tuning parameters associated with welfare evaluation ($h$) and with policy design ($k$) must be selected jointly. A bandwidth chosen in isolation—for example, by standard mean-squared-error bandwidth rules for welfare estimation—ignores its impact on the downstream optimization problem and is therefore generally sub-optimal. 

%This discussion highlights a key challenge in continuous-treatment policy learning: the tuning parameters from the evaluation and policy design stages must be tuned simultaneously. If we tune the bandwidth to optimize the welfare estimation error using standard bandwidth selection technique, and ignore the policy design stage, this would lead to a suboptimal choice of the tuning parameters.

If one instead fixes $h$ and maximizes the smoothed welfare $W_h$, the bias term vanishes and the model-selection method of \cite{mbakop2021model} applies directly. Our goal, however, is to optimize the actual welfare $W$, which requires balancing variance and kernel bias by adapting $h$ in tandem with the complexity parameter $k$.

%If we fix the bandwidth $h$ and shift our focus to maximizing $W_h$, the kernel bias term is eliminated, allowing us to proceed with the method developed by \cite{mbakop2021model}. However, when aiming to maximize $W$, it is essential to adjust the bandwidth $h$ to balance its effect on variance and kernel bias, which is the focus of our paper.

\section{Learning with Known Propensity Score: IPW} \label{sec:ipw}
This section examines the case in which the propensity score is known to the econometrician. Section \ref{sec:dr} addresses the case where the propensity score is unknown.

\subsection{Implementation} \label{sec:implementation}

The previous welfare decomposition suggests the following data-automated algorithm. Let $\mathcal{H} \subset (0,1)$ be a countable grid of bandwidth. First, we estimate the optimal policy for each pair $(h,k)$ as
\begin{align} \label{eqn:sieve-EWM}
    \hat{\pi}_{h,k} \equiv \argmax_{\pi\in\Pi_k} \hat{W}_h(\pi).
\end{align}
Then, we construct the penalized welfare as
\begin{align} \label{eqn:penalized-welfare-def}
    \hat{Q}_{h,k} \equiv \hat{W}_{h}(\hat{\pi}_{h,k}) - \underbrace{(\hat{R}_{h,k} + \tau(h,k,n) + B(h))}_{\text{penalty}},
\end{align}
where the three terms in the penalty are defined and explained in detail below.
The selector for $h$ and $k$ is the maximizer of $\hat{Q}_{h,k}$:\footnote{In practice, when the exact $\operatorname{VC}(\Pi_k)$ is unknown, one can replace it with an upper bound $V_k \geq \operatorname{VC}(\Pi_k)$. This substitution slightly alters the form of the oracle inequalities in Theorems \ref{thm:fully} and \ref{thm:dr}, as the infimum is taken over $\{V_k \leq nh^2\}$, which may cover a smaller range of $k$ compared to $\{\operatorname{VC}(\Pi_k) \leq nh^2\}$. Nevertheless, as we demonstrate in Section \ref{sec:sieve}, for many widely used policy classes—including highly complex ones such as neural networks—there exist well-established results providing tight or nearly tight upper bounds on VC dimension. We thank an anonymous referee for highlighting this point.} 

\begin{align} \label{eqn:h-k-hat}
    (\hat{h},\hat{k}) \equiv \argmax_{\substack{h\in\mathcal{H} \\
    k: \operatorname{VC}(\Pi_k) \leq nh^2}} \hat{Q}_{h,k}.
\end{align}
The final policy estimator is $\hat{\pi} \equiv \hat{\pi}_{\hat{h},\hat{k}}$. 

Before introducing the penalty terms, we briefly comment on the reason behind the restriction $\operatorname{VC}(\Pi_k) \leq nh^2$ imposed in the search over $(h,k)$ in (\ref{eqn:h-k-hat}). First, scanning every $k\in\mathbb{N}_{+}$ is computationally unrealistic in practice; second, the continuous-treatment setting presents a technical hurdle discussed in Section \ref{sec:discussion}. However, this restriction is not substantive in large samples. As discussed in (\ref{eqn:unrestricted-h-k}), the optimal variance–bias trade-off itself dictates the relationship $\operatorname{VC}(\Pi_k)\asymp nh^{2r+1}, r \geq 1$, which lies safely within the region defined by $\operatorname{VC}(\Pi_k)\le nh^{2}$. Hence, the restriction is asymptotically non-binding. Appendix \ref{sec:k-infty} outlines a full search procedure over $k\in\mathbb{N}_{+}$ that provides theoretical regret guarantees.

%Before explaining the penalty terms, we briefly discuss the restriction $\operatorname{VC}(\Pi_k) \leq nh^2$ in the search of $(h,k)$ in (\ref{eqn:h-k-hat}). We do not search over all $k \in \mathbb{N}_+$ as first it is impractical to do so in practice, and second this is associated with a technical difficulty specific to the continuous treatment case, as explained in the theoretical challenge part in Section \ref{sec:discussion}. However, as discussed in (\ref{eqn:unrestricted-h-k}) and the paragraph around it, the condition $\operatorname{VC}(\Pi_k) \leq nh^2$ is asymptotically non-binding, and allows for the choice of optimal $h$-$k$ relationship that balances the variance and kernel bias in (\ref{eqn:welfare-decomposition}): $\operatorname{VC}(\Pi_k) \asymp nh^{2r+1}$. 

The \emph{first penalty term} $\hat{R}_{h,k}$ is set to be
\begin{align*}
    \hat{R}_{h,k} \equiv  \mathbb{E}\left[\sup_{\pi \in \Pi_k} \frac{2}{nh} \sum_{i=1}^n \text{Rad}_i K\left(\frac{T_i-\pi(X_i)}{h}\right)\frac{Y_i}{f(T_i| X_i)}\Bigg | S_n\right],
\end{align*}
where $\{\text{Rad}_i: 1 \leq i \leq n\}$ represents a sequence of computer-generated i.i.d. Rademacher variables, independent of the sample $S_n \equiv \{(Y_i,T_i,X_i): 1 \leq i \leq n\}$. The expectation is taken over the Rademacher variables and is computed through simulations. The term $\hat{R}_{h,k}$, known as the (empirical) Rademacher complexity, is commonly used to penalize overfitting in policy estimators \citep{bartlett2002model,mbakop2021model}. Another approach to constructing the overfitting penalty involves using the holdout method described in \cite{mbakop2021model}, which we present in Section \ref{sec:holdout}.

The \emph{second penalty term} $\tau(h,k,n)$ is a user-specified technical term that ensures the penalty grows sufficiently fast with $h$ and $k$. Requirements for $\tau$ and specific choices are provided in the theorems. In general, $\tau$ is of smaller order than the leading term in the oracle inequality. 

As noted earlier, fixing \( h \) and focusing on maximizing \( W_h \) allows us to implement the penalty as \( \hat{R}_{h,k} + \tau \) following the method in \cite{mbakop2021model}. However, when the objective is \( W \), the penalty \( \hat{R}_{h,k} + \tau \) does not penalize against large bandwidth, causing uncontrolled kernel bias. Therefore, we introduce a third penalty term to offset this bias in welfare estimation.

The \emph{third penalty term} $B(h)$ is set to be
\begin{align*}
    B(h) \equiv& B(h;r,V_{\mu}) \equiv \frac{1}{2 \pii} \int |1- K^{\operatorname{FT}}(h \xi)|  V_{\mu} |\xi|^{-(r+1)} d\xi,
\end{align*}
where $\mu$ is the conditional expectation function of $Y$ given $T$:
\begin{align} \label{eqn:mu-def}
    \mu(t) \equiv \mathbb{E}[Y|T=t],
\end{align}
and $V_{\mu}$ is the total variation of $\mu$, and $r$ is the order of smoothness of $\mu$.\footnote{Total variation is defined as $V_{\mu} \equiv \sup_{m \in \mathbb{N}} \sup_{t_0,\cdots,t_m \in \mathcal{T}} \sum_{j=0}^m |\mu(t_j) - \mu(t_{j-1})|$.} $K^{\operatorname{FT}}$ is the Fourier transform of $K$, i.e., $K^{\operatorname{FT}}(\xi) \equiv \int K(t) e^{\mathbf{i}t\xi}dt$, with $\mathbf{i} \equiv \sqrt{-1}$. To avoid confusion with policy $\pi$, we use $\pii$ to denote the mathematical constant, the ratio of a circle's circumference to its diameter, approximately 3.14159.

This bias bound builds on the work of \cite{schennach2020bias}, establishing a tight upper bound for the nonparametric bias in kernel estimation. In our case, as shown in Lemma \ref{lm:bias-bound}, $B(h) \asymp h^r$ provides an upper bound for $|W_{h}(\pi) - W(\pi)|$ for any $\pi$. The quantities $r$ and $V_{\mu}$, and subsequently $B(h)$, can be estimated using the following double-debiased modification of the procedure described in \cite{schennach2020bias}.

Denote $\mu^{\operatorname{FT}}(\xi) \equiv \int \mu(t)e^{\mathbf{i}\xi t}dt$ as the Fourier transformation of $\mu$. This estimand admits the following double-debiased estimator:\footnote{The double-debiased property of this estimator is demonstrated in the proof in the appendix.}
\begin{align} \label{eqn:mu_hat-FT}
    \hat{\mu}^{\operatorname{FT}}(\xi) \equiv \int \hat{\mu}(t) e^{\mathbf{i}t\xi} dt + \frac{1}{n} \sum_{i=1}^n \frac{Y_i - \hat{\mu}(T_i)}{\hat{f}_T(T_i)}e^{\mathbf{i}T_i\xi},
\end{align}
where $\hat{\mu}$ and $\hat{f}_T$ are estimators of the corresponding nuisance functions.
As shown in \cite{schennach2020bias}, applying Fourier transformation converts the order of smoothness $r$ into the exponent in the frequency domain, yielding the bound $|\mu^{\operatorname{FT}}(\xi)| \leq V_{\mu} |\xi|^{-(r+1)}$. Taking the logarithm of both sides results in a linear expression: $\log |\mu^{\operatorname{FT}}(\xi)| \leq \log V_{\mu} - (r+1) \log|\xi|$. Thus, we can estimate $V_{\mu}$ and $r$ by finding the tightest linear upper bound on $\log |\hat{\mu}^{\operatorname{FT}}(\xi)|$ as a function of $\log|\xi|$, described by the following minimization problem:
\begin{align} \label{eqn:r-hat}
    (\hat{V}_{\mu},\hat{r}) & \equiv \argmin_{(A,r) \in \mathcal{A}} \int_{0}^{\log^2 n} (\log A - (r+1) \lambda) d\lambda, \\
    \mathcal{A} & \equiv \{(A,r): A \geq 0, r \in \mathbb{N}_+, \log |\hat{\mu}^{\operatorname{FT}}(\xi)| \leq \log A - (r+1) \log|\xi| \text{ for } 0 \leq \log|\xi| \leq \log^2 n\}. \nonumber
\end{align}
Plugging these estimates into the bias bound leads to the estimator $\hat{B}(h) \equiv B(h;\hat{r},\hat{V}_{\mu})$. 

This concludes our description of the procedure for data-automation of $(h,k)$. Its statistical properties are introduced next.

\subsection{Large sample properties}\label{sec:large}

The following assumptions are maintained regarding the sieve policy classes, the data-generating process, and the kernel function.

\begin{assumption}[Unconfoundedness] \label{ass:unconfoundedness}
    $T \perp \{Y(t):t \in \mathcal{T}\} | X$. 
\end{assumption}

\begin{assumption}[Welfare deficiency] \label{ass:deficiency}
    The welfare deficiency $W^*(\Pi_\infty) - W^*(\Pi_k) \rightarrow 0$, as $k \rightarrow \infty$. 
\end{assumption}

% Let $m$ be the conditional expectation function of $Y$ given $T$ and $X$:
% \begin{align} \label{eqn:m-def}
%     m(t,x) \equiv \mathbb{E}[Y|T=t,X=x].
% \end{align}

\begin{assumption}[Boundedness] \label{ass:bounded}
\ \begin{enumerate}[label = (\arabic*)]
    \item The treatment $T$ is compactly supported. A constant $\underline{f}>0$ exists such that $f \geq \underline{f}$ almost surely. 
    \item A constant $M>0$ exists, such that $|Y(t)| \leq M$ for all $t$. Consequently, $\lVert m \rVert_\infty \leq M$.
\end{enumerate}
\end{assumption}

\begin{assumption}[Kernel] \label{ass:infinite-kernel}
    The Fourier transform of $K$, $K^{\operatorname{FT}}$, satisfies that $K^{\operatorname{FT}}(\xi) = 1$ in a neighborhood of the origin, and $K^{\operatorname{FT}}(\xi)<1$ elsewhere. The kernel is symmetric and of bounded variation $\kappa_2 \equiv \int K(v)^2 dv < \infty$, and $\bar{\kappa} \equiv \sup_v K(v) < \infty.$
\end{assumption}

\begin{assumption}[Smoothness: $r$th order]\label{ass:smoothness-r}
The functions $m(\cdot,x), x \in \mathcal{X}$, and $\mu(\cdot)$, as defined in (\ref{eqn:m-def}) and (\ref{eqn:mu-def}), satisfy the following smoothness conditions:
    \begin{enumerate}[label = (\arabic*)]
        \item For some $r \in \mathbb{N}_+$, $\mu(\cdot)$ and $m(\cdot,x)$ are $r$ times differentiable with the $r$th derivative absolutely continuous except over a finite non-empty set of points.
        \item For all $x \in \mathcal{X}$, the total variation of $m(\cdot,x)$ does not exceed that of $\mu(\cdot)$, i.e., $\sup_{x \in \mathcal{X}} V_{m(\cdot,x)} \leq V_{\mu}$.
    \end{enumerate}
\end{assumption}

\begin{assumption}[Bias bound estimation] \label{ass:bias-est}
Assume that the estimators $\hat{\mu}$ and $\hat{f}_T$ in the estimation of $B(h)$ are constructed using cross-fitting with a finite number of folds and satisfy the following conditions:\footnote{The cross-fitting procedure here can be implemented similarly to the approach described in Section \ref{sec:dr}, and is omitted for brevity.}
   \begin{enumerate}[label = (\arabic*)] 
        \item $\hat{f}_T$ bounded away from zero, 
        \item $\mathbb{E}[\int (\hat{\mu}(t) - \mu(t))^2 f_T(t) dt], \mathbb{E}[\int (\hat{f}_T(t) - f_T(t))^2 f_T(t) dt] \leq n^{-\epsilon}$ for some $\epsilon>0$,
        \item $\lVert \hat{\mu} - \mu \rVert_\infty \lVert \hat{f}_T - f_T \rVert_\infty = o_{a.s.}(n^{-1/2})$.
    \end{enumerate}
\end{assumption}

Assumption \ref{ass:unconfoundedness} establishes the identification of the welfare function, under which we can write $m(t,x) = \mathbb{E}[Y|T=t,X=x]$.
Assumption \ref{ass:deficiency} requires that the sequence of sieve policy spaces approaches the global target eventually. 
In Assumption \ref{ass:bounded}, we assume that both the outcome and the inverse propensity are bounded. The discrete-treatment version of this assumption is considered in \cite{kitagawa2018should} as Assumption 2.1 and in \cite{mbakop2021model} as Assumption 3.1.\footnote{\cite{athey2021policy} do not require the outcome to be bounded (only requiring its distribution to exhibit sub-Gaussian tail). Still, they do maintain the requirement for the inverse propensity score to be bounded.}

Assumption \ref{ass:infinite-kernel} specifies that the kernel is of infinite order. %One example of an infinite-order kernel is sinc kernel $K(v) = \sin(v)/(\pi v)$. 
See \cite{devroye1992note,politis1999multivariate} for examples and discussions. The use of an infinite-order kernel is not essential. Any sufficiently high-order kernel that accommodates the smoothness of the dose-response function would be effective. Assumption \ref{ass:smoothness-r} is adopted from \cite{schennach2020bias} and essentially states that the relationship from the treatment to outcome is $r$th-order smooth uniformly over the covariates. Assumption \ref{ass:bias-est} imposes standard conditions on the nuisance estimators used in constructing the bias bound, ensuring its consistency.

\begin{theorem}\label{thm:fully}
Let Assumptions \ref{ass:unconfoundedness} - \ref{ass:bias-est} hold.
The bandwidth grid $\mathcal{H}$ and the technical term $\tau$ satisfies that $\tau(h,k,n)\in(0,1)$, and for any constant $C>0$,
\begin{align} \label{eqn:tau-requirement}
\sum_{k=1}^\infty\sum_{h\in\mathcal{H}}\exp(-nh\tau(h,k,n)^2)/C)
\end{align}
is finite and stays bounded as $n\rightarrow\infty$.

\noindent (1) If $B(h)$ is known, and we set $h_{{min}} \equiv \inf \mathcal{H} \gtrsim n^{-1/(2r+1)}$,\footnote{For any two sequences $a_n$ and $b_n$, $a_n \gtrsim b_n$ means that there exists a constant $c >0$ such that $a_n \geq c b_n$.} then the following oracle inequality holds,
\begin{align*}
    & W^*(\Pi_\infty) - W(\hat{\pi}) \nonumber \\
    \leq &\inf_{\substack{h\in\mathcal{H} \\ k:\operatorname{VC}(\Pi_k)\leq nh^2}}\left(W^*(\Pi_\infty)-W^*(\Pi_k)+2(C_v+o(1))\sqrt{\frac{\operatorname{VC}(\Pi_k)}{nh}}+2B(h)+\tau(h,k,n)\right) \\
    & + O_p(n^{-r/(2r+1)}), \nonumber
\end{align*}
where $C_{v} \equiv cM\sqrt{\frac{\kappa_2}{\underline{f}}}$ (the subscript $v$ denotes variance), with $c$ being a universal constant that can be computed explicitly as detailed in the proof.

\noindent (2) If $B(h)$ is unknown, then in the construction of $\hat{Q}_{h,k}$ in (\ref{eqn:penalized-welfare-def}), we replace $B(h)$ by $(1+\gamma)\hat{B}(h)$ for an arbitrarily small $\gamma>0$ and set $h_{{min}} \gtrsim n^{-1/(2\hat{r}+1)}$, then the following oracle inequality holds,
\begin{align*}
    & W^*(\Pi_\infty) - W(\hat{\pi}) \nonumber \\
    \leq & \inf_{\substack{h\in\mathcal{H} \\ k:\operatorname{VC}(\Pi_k)\leq nh^2}} \Bigg(W^*(\Pi_\infty)-W^*(\Pi_k)+2(C_v+o(1))\sqrt{\frac{\operatorname{VC}(\Pi_k)}{nh}}+2(1+\gamma)B(h) +\tau(h,k,n)\Bigg) \\
    & + O_p(n^{-r/(2r+1)}). \nonumber
\end{align*}
\end{theorem}

The first part of Theorem \ref{thm:fully} describes the infeasible performance of the policy estimator when the bias bound is known, while the second part presents an oracle inequality when the bias bound is consistently estimated. When the bias bound is estimated, it needs to be inflated by a factor of $1+\gamma$ to ensure it correctly bounds the bias with high probability. In both cases, the procedure can balance the tradeoff among the three terms in the welfare decomposition of (\ref{eqn:welfare-decomposition}), introducing a technical term and a small order term $O_p(n^{-1/(2r+1)})$.\footnote{We use an $O_p$ remainder rather than a nonasymptotic high-probability bound in the oracle inequalities due to the estimation of $r$. See the proof of Theorem \ref{thm:fully}(2) for details.}
 The technical term can be chosen to be dominated by the variance term, as described subsequently.

The term $\tau$ and condition (\ref{eqn:tau-requirement}) are inevitable artifacts of applying a union bound to control the random selectors $(\hat h,\hat k)$ in (\ref{eqn:union-bound}) when proving Theorem \ref{thm:fully}. This proof device—and the resulting technical term—is standard (e.g., \citealp{mbakop2021model,bartlett2002model,koltchinskii2001rademacher,koltchinskii2011oracle,boucheron2005theory}); our setting is even more delicate because it involves two tuning parameters rather than one. Condition (\ref{eqn:tau-requirement}) imposes a lower bound on $\tau$ (since $nh\tau^{2}$ must be large), but $\tau$ must also remain small relative to the terms in the oracle inequality. Below we provide specific choices of $\tau$, based on feasible choices of the bandwidth grid, to satisfy (\ref{eqn:tau-requirement}) while staying negligible compared to $\sqrt{\operatorname{VC}(\Pi_k)/(nh)}$:
\begin{align*}
    \text{exponential sequence: } \mathcal{H} &= \left\{h: h= \rho^{-j},  j\in\mathbb{N}, h\geq n^{-1/(2\hat{r}+1)}\right\}, \rho>1, \\
    \text{geometric sequence: } \mathcal{H} &= \left\{h: h= j^{-\rho}, j\in\mathbb{N}, h\geq n^{-1/(2\hat{r}+1)}\right\}, \rho >0. \\
    \tau(h,k,n)& = \sqrt{\frac{\lambda_k\log k -\lambda_h', \log h}{nh}},
\end{align*}
for any sequences $\lambda_k\uparrow \infty$ as $k \rightarrow \infty$ and $\lambda_h' \uparrow \infty$ as $h \rightarrow 0$. 
It is worth noting that when the global policy class has a finite VC dimension—e.g., in Example \ref{eg:piecewise_linear} with $K<\infty$—the technical term reduces to $\tau(h,n)$; the dependence on $k$ disappears because the sum over $k$ is finite. If the bandwidth grid $\mathcal{H}$ is also finite, this term is no longer needed at all.

In large samples, the constraint $\operatorname{VC}(\Pi_k) \le nh^{2}$ in the oracle inequality’s infimum is asymptotically non-restrictive. Let $(h_n^{*},k_n^{*})$ solve the unconstrained problem:
\begin{align*}
    \inf_{h\in\mathcal{H}, k \geq 1} \Bigg(W^*(\Pi_\infty)-W^*(\Pi_k)+2(C_v+o(1))\sqrt{\frac{\operatorname{VC}(\Pi_k)}{nh}}+2(1+\gamma)B(h) +\tau(h,k,n)\Bigg).
\end{align*}
Then
\begin{align} \label{eqn:unrestricted-h-k}
    \sqrt{\frac{\operatorname{VC}(\Pi_{k_n^{*}})}{n\,h_n^{*}}}\asymp(h_n^{*})^{r},
\quad\text{so}\quad
\frac{\operatorname{VC}(\Pi_{k_n^{*}})}{n(h_n^{*})^{2}}
   =\left(\sqrt{\tfrac{\operatorname{VC}(\Pi_{k_n^{*}})}{n\,h_n^{*}}}\right)^{2}(h_n^{*})^{-1}
   \asymp(h_n^{*})^{2r-1}\to0, r \geq 1.
\end{align}
Hence, the unconstrained optimal choice automatically satisfies $ \operatorname{VC}(\Pi_{k_n^{*}}) \le n(h_n^{*})^{2}$, rendering the restriction non-binding.

From (\ref{eqn:unrestricted-h-k}), we can obtain the rate for the unconstrained optimal bandwidth 
\begin{align*}
    h_n^* \asymp \left( \frac{\operatorname{VC}(\Pi_{k_n^{*}})}{n\,h_n^{*}} \right)^{\frac{1}{2r+1}}.
\end{align*}
Because $\Pi_\infty$ is infinite-dimensional, driving the welfare deficiency to zero requires $k_n^*\to\infty$, so $h_n^*$ is of larger order than the standard optimal $n^{-1/(2r+1)}$ rate for nonparametric kernel regression. This justifies setting $h_{\min}\gtrsim n^{-1/(2r+1)}$ and then searching upward from there. 

%: the optimal bandwidth here must exceed the standard optimal bandwidth rate in nonparametric regression (which ignores policy learning and sieve approximation). We therefore start the data-driven search at $n^{-1/(2r+1)}$ and increase from there.

\subsection{Discussion} \label{sec:discussion}

\paragraph{Comparison with the literature}

Although the work of \cite{athey2021policy} addressed continuous treatment, their approach examines infinitesimal nudges through outcome derivatives, effectively reducing the problem to a binary treatment scenario. They did not consider a sieve approximation of the policy space.

In contrast, \cite{mbakop2021model} considered sieve approximation for the binary treatment case. Their Proposition 3.2 derives the IPW (with known propensity score) welfare regret bound, which is
\begin{align*}
    \inf_{k \geq 1} \left( W^*(\Pi_\infty) - W^*(\Pi_k) + C \sqrt{\frac{\operatorname{VC}(\Pi_k)}{n}} + \sqrt{\frac{k}{n}} \right) + O\left( \frac{1}{\sqrt{n}}\right).
\end{align*}
Comparing our bound in Theorem \ref{thm:fully} to \cite{mbakop2021model}'s bound, the differences are as follows. First, our bound includes an additional bias term of order \( h^r \), and the variance term is larger by a factor of \( 1/\sqrt{h} \). This difference arises from the nonparametric kernel estimation required for the continuous treatment. Similarly, the order of the remainder term increases from $n^{-1/2}$ to \( n^{-r/(2r+1)} \), reflecting the standard minimax rate of nonparametric estimation under smoothness $r$ without approximation of the policy space. Additionally, the technical term has been adjusted to account for the data-driven selection of bandwidth.\footnote{In fact, as shown by \cite{bartlett2002model}, in the binary treatment case, the technical term in the binary treatment case can be made as small as $\sqrt{(\log k)/n}$ instead $\sqrt{k/n}$.} 

Second, our bound directly applies to the regret itself, whereas \cite{mbakop2021model}'s bound applies to the expected regret \(\mathbb{E}[W^*(\Pi_\infty) - W(\hat{\pi})]\). That is, we additionally bound (in probability) the deviation of the regret from its mean. However, it is important to note that our result does not imply convergence in expectation, as the \(O_p\)-remainder terms in Theorem \ref{thm:fully} may only converge in the (weaker) notation of convergence probability. This limitation is due to the complexities in estimating the order of smoothness, which typically exhibits a slow convergence rate.\footnote{See, for example, \cite{sun2005adaptive}, the minimax rate for the order of smoothness is logarithmic in $n$.}

\cite{kallus2018policy} studied policy learning with continuous treatment, obtaining a regret bound of order $\mathcal{R}_n(\Pi_\infty)/h^2  + \text{bias}(h)$,
where $\mathcal{R}_n(\Pi_\infty)$ is the Rademacher complexity of $\Pi_\infty$, defined as $\mathcal{R}_n(\Pi_\infty) \equiv \mathbb{E}\left[ \sup_{\pi \in \Pi_\infty} 2\sum_{i=1}^n \text{Rad}_i \pi(X_i)/n \right]$.
There are several key differences between their result and ours. First, they do not consider sieve approximation of the policy space. Second, even without sieve approximation, the bound derived by \cite{kallus2018policy} is less sharp than our Theorem \ref{thm:fully}.  This can be seen in a global policy class with a finite VC dimension $\Pi_\infty = \Pi_k$. In this more straightforward scenario, both procedures aim to find the optimal policy within a fixed policy class of finite VC dimension, leading to a welfare deficiency of zero. The resulting regret bounds are respectively
\begin{align*}
    \text{\cite{kallus2018policy}: }  C\sqrt{\frac{\operatorname{VC}(\Pi_\infty)}{nh^4}} + \text{bias}(h), \quad 
    \text{our Theorem \ref{thm:fully}: } C\sqrt{\frac{\operatorname{VC}(\Pi_\infty)}{nh}} + \text{bias}(h).
\end{align*}
Observe that our variance term is smaller by a factor of $h^{-3/2}$. This difference arises because in the analysis of \cite{kallus2018policy}, they only utilize the Rademacher contraction comparison lemma \citep[e.g., Theorem 4.12 in][]{ledoux1991probability} to obtain the bound. However, this approach only leverages the contraction property of the function $hK(\cdot/h)$ and does not fully exploit the following structure of the kernel term:
\begin{align} \label{eqn:kernel-property}
\begin{split}
       \sup_{t,x,\pi} \left| \frac{1}{h} K\left( \frac{t-\pi(x)}{h} \right) \right|  \lesssim \frac{1}{h}, ~
    \sup_{\pi} \mathbb{E}\left[ \left| \frac{1}{h} K\left( \frac{t-\pi(x)}{h} \right) \right|^2 \right]  \lesssim \frac{1}{h}. 
\end{split}
\end{align}
For the kernel term, this implies that although it is uniformly bounded by \( 1/h \), its second moment is also bounded by \( 1/h \) rather than \( 1/h^2 \) due to a standard change of variables argument. If we apply uniform bounds, which are sufficiently sharp in discrete-treatment policy learning, the variance term becomes \( \frac{1}{h}\sqrt{\frac{\text{VC}(\Pi_k)}{n}} \). Thus, using second-moment-type bounds to estimate \( \mathbb{E}[\hat{R}_{h,k}] \) yields significantly sharper results than relying on uniform bounds. Unlike our approach, \cite{kallus2018policy} did not leverage this ``small second-moment property'' of the kernel term, leading to a much larger bound on the regret.

\paragraph{Theoretical challenges} 
In connection with the kernel properties in (\ref{eqn:kernel-property}), establishing that the penalized welfare closely approximates the true welfare—analogous to Lemma 3.2 in \cite{mbakop2021model}—requires a sharper concentration result. Specifically, we rely on Talagrand’s inequality \citep{talagrand1996new} to control the tail probability, rather than the bounded difference inequality \citep{mcdiarmid1989method}. The use of Talagrand’s inequality is essential here, as it incorporates the second moment of the empirical process and avoids the overly conservative bounds that arise from the bounded difference inequality, which relies solely on a uniform bound for the kernel term.

However, Talagrand’s inequality introduces an additional term in the denominator, which grows with the sieve index. To ensure that this term does not dominate, we impose the technical condition \(\operatorname{VC}(\Pi_k) \leq nh^2\) in our procedure. As discussed in Appendix \ref{sec:k-infty}, increasing \(\tau\) offers an alternative solution to this technical issue. Nonetheless, this remains a technical concern, as practically searching over policy spaces with complexities up to \(nh^2\) is more than adequate.

In contrast, in the discrete case, it suffices to use the bounded difference inequality, as in Lemma 3.2 of \cite{mbakop2021model}, because a uniform bound on the empirical IPW welfare yields a tail bound that is as sharp as one based on the second moment.

\subsection{Holdout penalty} \label{sec:holdout}
As noted in \cite{mbakop2021model}, the holdout penalty is an alternative to using Rademacher complexity for penalizing overfitting. This approach involves setting aside a portion of the sample to assess the performance of sieve policy estimators.

The holdout procedure is described below. Let $n_E=\lfloor (1-\iota) n \rfloor$ and $n_T = n - n_E$ for some fixed number $\iota\in(0,1)$.\footnote{For concreteness, one can consider $\iota = 1/2, 1/4$ as in the numerical examples of \cite{mbakop2021model}.} The original sample $S_n=\{(Y_i,D_i,X_i)\}_{i=1}^n$ is split into an estimating sample $S_{n_E}\equiv\{(Y_i,D_i,X_i)\}_{i=1}^{n_E}$ and a testing sample $S_{n_T}\equiv\{(Y_i,D_i,X_i)\}_{i=n_E+1}^n$. The estimating sample is used to identify the optimal policy within each \( \Pi_k \), while the testing sample evaluates the performance of the estimated policy. Let \( \hat{W}^E_{h}(\pi) \) represent the empirical welfare on the estimating sample and \( \hat{\pi}^E_{h,k} \equiv \arg\max_{\pi \in \Pi_k} \hat{W}^E_{h}(\pi) \) be the empirical welfare maximizer. Let \( \hat{W}^T_{h} \) denote the empirical welfare obtained on the testing sample.

The holdout procedure constructs the penalized welfare as 
\begin{align*}
    \hat{Q}^{\textit{hold}}_{h,k} & \equiv \hat{W}^E_{h}(\hat{\pi}^E_{h,k}) - \underbrace{\left( \hat{W}^E_{h}(\hat{\pi}^E_{h,k}) - \hat{W}^T_{h}(\hat{\pi}^E_{h,k}) + \tau(k,h,n) + B(h) \right)}_{\text{penalty}} \\
    & = \hat{W}^T_{h}(\hat{\pi}^E_{h,k}) - B(h) - \tau(k,h,n).
\end{align*}
Intuitively, this procedure penalizes overfitting using the difference between in-sample and out-of-sample estimated welfare. The procedure remains the same as in Section \ref{sec:implementation}.

\begin{corollary} \label{cor:holdout}
    Given the assumptions of Theorem \ref{thm:fully}, the holdout procedure achieves the same oracle inequality, but with the sample size \( n \) replaced by the size of the estimating sample \( n_E \).
\end{corollary}

The Rademacher penalty, which exploits the full sample, yields lower regret but is harder to analyze theoretically. In contrast, the holdout penalty is theoretically simpler because the independence of the held-out sample allows the use of Bernstein’s inequality, without requiring the more technically involved Talagrand’s inequality. Yet, the holdout method sacrifices efficiency because the held-out data are not directly used to estimate the policy function.

%Using the holdout penalty, we circumvent the theoretical difficulties posed by Talagrand's inequality, as Bernstein’s inequality is sufficient due to the independence of the held-out sample. However, the holdout method only leverages a portion of the sample, unlike Rademacher complexity, which uses the whole sample. This distinction between the two methods becomes prominent only in the continuous treatment setting. 

\section{Learning with Observational Data: Double Debias}\label{sec:dr}

In the previous section, we considered the case where the propensity density $f$ is known. This section studies the case where the propensity density is unknown. Using double debiasing techniques, we can achieve a welfare regret bound similar to the one derived in the previous section.

Define the double-debiased welfare function as
\begin{align} \label{eqn:double-robustness}
    \Gamma_h(Y,T,X;\pi;f,m) \equiv & \frac{1}{h} K \left( \frac{T-\pi(X)}{h} \right) \frac{Y - m(\pi(X),X)}{f(T|X)} + m(\pi(X),X) \\
    = & \underbrace{\frac{1}{h} K \left( \frac{T-\pi(X)}{h} \right) \frac{Y }{f(T|X)}}_{\text{IPW moment function}} + \underbrace{\left(1- \frac{1}{hf(T|X)} K \left( \frac{T-\pi(X)}{h} \right)  \right)m(\pi(X),X)}_{\text{adjustment term}}. \nonumber
\end{align}
This expression represents the IPW welfare studied in Section \ref{sec:ipw} with an added adjustment term. The adjustment term has zero mean, i.e., \( \mathbb{E}[\Gamma_h] = W_h \), but introduces additional variability. To manage this added variance, we impose the following assumption.

\begin{assumption} \label{ass:m-bounded-variation}
    For each $x$, $m(\cdot,x)$ is of bounded variation.
\end{assumption}

Assumption \ref{ass:m-bounded-variation} places a mild requirement on the dose-response function $m$. By restricting $m$ to be of bounded variation in the first argument, we can bound the complexity of the function class
\begin{align*}
    \{x \mapsto m(\pi(x),x): \pi \in \Pi_k\}
\end{align*}
by using the VC dimension of $\Pi_k$, thus controlling the additional variance introduced by the adjustment term in (\ref{eqn:double-robustness}). This technique is standard in nonparametric estimation. For example, \cite{GINE2002rates} uses this condition on the kernel function to ensure that the function class corresponding to the kernel density estimator is of finite VC dimension. 

An alternative approach to formulating the double-debiased moment function involves modifying $m(\pi(X),X)$ instead of the IPW expression:\footnote{\cite{kallus2018policy} briefly mentioned, without providing formal results, the double-debiased approach using this formulation instead of our (\ref{eqn:double-robustness}).}
\begin{align*}
& m(\pi(X),X) + \frac{1}{h} K \left( \frac{T-\pi(X)}{h} \right) \frac{1}{f(T|X)}(Y-m(T,X)). %\\ = & \underbrace{\frac{1}{h} K \left( \frac{T-\pi(X)}{h} \right) \frac{1}{f(T|X)}}_{\text{IPW moment function}} + \underbrace{m(\pi(X),X) - \frac{1}{h} K \left( \frac{T-\pi(X)}{h} \right) \frac{m(T,X)}{f(T|X)}}_{\text{adjustment term}}.
\end{align*}
In this case, the adjustment term introduces additional estimation bias beyond that of the IPW estimand, which is difficult to control. For this reason, we do not pursue this version of the double-debiased formula.

Estimating the welfare based on the double-debiased moment function in (\ref{eqn:double-robustness}) requires first-stage estimation of $f$ and $m$. For notation simplicity, we will use $g \equiv 1/f$ to denote the inverse propensity. 
Suppose we have consistent (under the sup-norm) estimators $\hat{m}$ and $\hat{g}$ for $m$ and $g$, respectively. They are assumed to satisfy the following conditions.
\begin{assumption} \label{ass:rates-nonparametric-estimates}
    The estimators $\hat{g}$ and $\hat{m}$ satisfy the following conditions:
    \begin{enumerate}[label = (\roman*),noitemsep]
       
        \item There exist $\rho_g,\rho_m \geq \frac{r}{2(2r+1)}$ 
        such that $\lVert \hat{g}-g \rVert_{\infty} = o_p(n^{-\rho_g})$ and $\lVert \hat{m}-m \rVert_{\infty}= o_p(n^{-\rho_m})$.
        \item With probability approaching one, $\hat{g}$ and $\hat{m}$ are bounded.
        \item With probability approaching one, $\hat{m}(\cdot,x)$ is of bounded variation for each $x$.
        \end{enumerate}
\end{assumption}

Assumption \ref{ass:rates-nonparametric-estimates}(i) concerns the mean-squared convergence rate of \( \hat{g} \) and \( \hat{m} \) in the \( L_\infty \) space.\footnote{Alternatively, we could impose assumptions of convergence rate on $L_2$ norm for $X$, but infinity norm for $T$ is still needed. Previous literature assumes $L_2$ convergence rate for discrete treatment, while we consider continuous treatment in contrast, and hence, we need the rate to hold uniformly for all treatment levels.} Unlike in \cite{athey2021policy}, these rate constraints depend on the smoothness \( r \), as the remainder term in the welfare regret is of order \( n^{-r/(2r+1)} \) rather than \( n^{-1/2} \). Since \( r \) is unknown, one could either use the estimate \( \hat{r} \) to assess this condition or adopt the conservative choice of having \( \rho_g, \rho_m \geq 1/4 \).\footnote{It is difficult to directly compare our rate requirement on the nuisance estimators with that of \cite{athey2021policy} because of the additional continuous argument \( T \) in the nuisance functions.} The other conditions in Assumption \ref{ass:rates-nonparametric-estimates} require that the estimators share the same properties as their target.

%As in the recent work on double debias methods \citep{athey2021policy,chernozhukov2018double}, we take an agnostic view on how the nuisance estimates $\hat{g}$ and $\hat{m}$ are obtained and impose high-level conditions on their rates of convergence. 

Given sufficient regularity, we can construct an estimator of $m$ that satisfies the rate condition in Assumption \ref{ass:rates-nonparametric-estimates} by employing, for example, sieve-based methods \citep{chen2007large}, local polynomial methods \citep{calonico2018effect}, or modern machine-learning techniques such as random forests, lasso, ridge, deep neural nets, boosted trees, and ensembles of these methods \citep{chernozhukov2018double}. We can also use recent advances in linear and nonlinear partitioning-based methods \citep{CattaneoFarrellFeng2020large,CattaneoChandakKlusowski2024convergence,CattaneoFengShigida2024uniform}, which encompass certain decision-tree and recursive-partitioning approaches. To estimate the conditional density $f$, one may apply the techniques developed by \citet{cattaneo2024boundary} or \citet{ColangeloLee2025}. Appendix \ref{sec:nonparametric-estimators} provides guidance on constructing the nonparametric estimators $\hat g$ and $\hat m$ based on \citet{cattaneo2024boundary} and \citet{CattaneoFengShigida2024uniform} and outlines sufficient conditions under which these estimators satisfy Assumption \ref{ass:rates-nonparametric-estimates}.

For welfare estimation, we implement the following cross-fitting procedure. Divide the data equally into $L$ folds, using the size of each fold $n/L$. For $\ell = 1, \cdots, L$, let $I_\ell$ denote the set of observation indices in the $\ell$th fold and $I^c_\ell = \bigcup_{\ell' \ne \ell} I_{\ell'}$ the set of observation indices not in the $\ell$th fold. With a slight abuse of notation, denote $S_\ell$ as the set of observations with indices $i \in I_\ell$. For observation $(T_i,X_i)$ in $S_\ell$, we use the observations with indices in $I^c_\ell$ to construct the nonparametric estimators $\hat{m}_\ell(T_i,X_i)$ and $\hat{g}_\ell(T_i,X_i)$. The subscript $\ell$ signifies that the estimators are constructed using data in $I_\ell^c$.
The double-debiased empirical welfare is constructed as 
\begin{align*}
    \hat{W}_{h}^{\operatorname{DD}}(\pi) \equiv \frac{1}{n} \sum_{\ell =1}^L \sum_{i \in I_\ell}  \Gamma_{h_{k,n}}(Y_i,T_i,X_i;\pi;\hat{g}_{\ell},\hat{m}_\ell),
\end{align*}
where the superscript DD indicates double debias.
The sieve empirical welfare maximizer with double-debiased welfare is obtained as 
\begin{align*}
    & \hat{\pi}_{h,k}^{\operatorname{DD}} \equiv \arg\max_{\pi \in \Pi_k} \hat{W}_{h}^{\operatorname{DD}}(\pi).
\end{align*}
The penalized welfare is set to be
\begin{align*}
    \hat{Q}_{h,k}^{\operatorname{DD}}\equiv \hat{W}_{h}^{\operatorname{DD}}(\hat{\pi}_{h,k}^{\operatorname{DD}}) - \left( \frac{1}{L} \sum_{\ell=1}^L \hat{R}_{h,k}^{\operatorname{DD},\ell}  + \tau(h,k,n) + (1+\gamma)\hat{B}(h) \right),
\end{align*}
where the Rademacher complexity is now computed using the double-debiased moment with cross-fitting
\begin{align*}
    \hat{R}_{h,k}^{\operatorname{DD},\ell} \equiv \mathbb{E} \left[ \sup_{\pi \in \Pi_k} \frac{L}{n} \sum_{i \in I_\ell} 2\text{Rad}_i \cdot \Gamma_h(Y_i,T_i,X_i;\pi;\hat{g}_{\ell},\hat{m}_{\ell}) \Bigg| S_\ell \right].
\end{align*}
The rest of the procedure remains the same as in Section \ref{sec:implementation}, and we denote the resulting policy estimator as $\hat{\pi}^{\operatorname{DD}}$.

\begin{theorem}\label{thm:dr}
Let assumptions of Theorem \ref{thm:fully}(2) and Assumptions \ref{ass:m-bounded-variation} and \ref{ass:rates-nonparametric-estimates} hold. If $h_{{min}} \gtrsim n^{-1/(2\hat{r}+1)}$, then the following oracle inequality holds for the policy estimator $\hat{\pi}^{\operatorname{DD}}$:
\begin{align*}
    & W^*(\Pi_\infty) - W(\hat{\pi}^{\operatorname{DD}}) \nonumber \\
    \leq & \inf_{\substack{h\in\mathcal{H} \\ k:\operatorname{VC}(\Pi_k)\leq nh^2}} \Bigg(W^*(\Pi_\infty)-W^*(\Pi_k)+2(C_v'+o(1))\sqrt{L\frac{\operatorname{VC}(\Pi_k)}{nh}}+2(1+\gamma+o(1))B(h) +\tau(h,k,n)\Bigg) \\
    & + O_p(n^{-r/(2r+1)}),
\end{align*}
where $C_v' \equiv (c+c')M \sqrt{\frac{\kappa_2}{\underline{f}}}$, with $c'$ being a universal constant different from $c$, as specified in the proof.
\end{theorem}

Comparing Theorem \ref{thm:dr} with Theorem \ref{thm:fully}, the variance bound is higher in two ways. First, the constant \( C_v' > C_v \) accounts for the extra variation introduced by the double-debiased adjustment term in (\ref{eqn:double-robustness}). Second, there is a factor \( \sqrt{L} \) due to the Rademacher complexity being constructed through cross-fitting; this factor also appears in \cite{zhou2023offline} (e.g., their Lemma 3). 

The bias term is slightly larger because, although the double-debiasing procedure mitigates bias, a small residual remains owing to kernel smoothing of the continuous treatment. In standard double-debiasing calculations for discrete treatments, the bias includes the expectation of the following term
\begin{align*}
    \mathbf{1}\{T=\pi(X)\}\,(Y-m(\pi(X),X))\bigl(\hat g_\ell(T,X)-g(T,X)\bigr),
\end{align*}
whose expectation is zero. With a continuous treatment, this term becomes
\begin{align*}
    \frac{1}{h}K\!\left(\frac{T-\pi(X)}{h}\right)(Y-m(\pi(X),X))\bigl(\hat g_\ell(T,X)-g(T,X)\bigr),
\end{align*}
whose expectation is generally non-zero and depends on the bandwidth $h$ and the convergence rate of $\hat{g}$. See Lemma \ref{lm:DR-calculation} for details.

Despite these two differences, the result for the observational setting shows that double-debiased policy estimators can achieve a comparable welfare regret to the IPW setting with a known propensity, provided the first-stage estimates converge at a sufficiently fast rate.

\section{Sieve Policy Class Construction}\label{sec:sieve}
In this section, we discuss the implementation of sieve approximation of the global policy class, including traditional sieves and neural networks. In each case, we discuss the VC dimension and the welfare approximation rate for the specific sieve class. 

Specifically, let $\alpha_k \downarrow 0$ denote the rate at which the sieve sequence $\Pi_\infty$ approximates $\Pi_k$, which is defined by 
\begin{align} \label{eqn:deficiency-rate}
\sup_{\pi\in\Pi_\infty}\inf_{\pi_k\in\Pi_k}\mathbb{E}|\pi(X)-\pi_k(X)| \leq \alpha_k.
\end{align}
If $m$ is Lipschitz in $t$ with a Lipschitz constant uniform in $x$ (e.g., its first derivative in $t$ is bounded), then the welfare deficiency $W^*(\Pi_\infty) - W^*(\Pi_k)$ is $O(\alpha_k)$. When $W^*(\Pi_\infty)$ is achieved at an interior maximizer satisfying the first-order condition, the welfare deficiency rate can be sharpened to $O(\alpha_k^2)$. We present upper bounds on welfare deficiency only for illustration; the oracle inequalities, as shown in the theorems, ensure that the data-driven estimator automatically balances variance and kernel bias against the actual welfare deficiency without needing to know its exact rate.\footnote{Similar data-driven sieve-selection ideas have been explored in other adaptive estimation problems (e.g., \citet{breunig2024adaptive} and \citet{chen2025adaptive} for nonparametric instrumental variables), though the goals there differ from policy learning.}

\subsection{Monotone policies} \label{sec:monotone-sieve}

As in the empirical approach of \cite{mbakop2021model}, a standard way for restricting \( \Pi_\infty \) is to apply shape constraints driven by economic principles, such as fairness. We introduce formulations of \( \Pi_\infty \) and \( \Pi_k \) that extend the structure used in the empirical study of \cite{mbakop2021model} to our continuous-treatment context. In Section \ref{sec:applications}, this policy class formulation promotes fairness within job training programs.

Denote the $p$-th component of $X$ by $X_p$ and the support of $X_p$ by $\mathcal{X}_p\subset \mathbb{R}$. Let $h_p:\mathcal{X}_p\rightarrow \mathbb{R}$ be a monotone function and Lipschitz continuous with constant $L_p$, for $1\leq p \leq d_X$. We consider a particular type of policy, which transforms each coordinate of $X$ and then takes the sum as the treatment level: $\pi(X) = \sum_{p=1}^{d_X}h_p(X_p)$.
Let $\Pi_{\infty}$ be the set of all such policies: 
\begin{align*}
\Pi_{\infty} \equiv \left\{\pi(x)=\sum_{p=1}^{d_X}h_p(x_p): h_p \text{ increasing and bounded},\forall p \right\}.
\end{align*}
To construct the sieve policy class $\Pi_k$, we define $\varphi_{k,k^\prime}$ as the triangular kernel shifted by $k^\prime$ and scaled by $k$, for $k'=1,\dots,k$:
\begin{align*}%\label{eqn:varphi}
\varphi_{k,k^\prime}(x_p)\equiv \mathbf{1}\left\{\frac{k^\prime-1}{k}\leq x_p \leq \frac{k^\prime+1}{k}, 0 \leq x \leq 1\right\}\left(1 - |kx_p - k^\prime|\right), 0 \leq k' \leq k.
\end{align*}
Given a vector of coefficients $\theta \equiv (\theta_1,\cdots,\theta_{d_X})$, where $\theta_p\equiv (\theta_{p,0},\cdots,\theta_{p,k})^\top$ for $p=1,\dots,d_X$, the policy $\pi_{k,\theta}$ is defined as
\[\pi_{k,\theta}(x)\equiv \sum_{p=1}^{d_X}\sum_{k^\prime=0}^k \theta_{p,k^\prime}\varphi_{k,k^\prime}(x_p).\]
To ensure that $\sum_{k^\prime=0}^k \theta_{p,k^\prime}\varphi_{k,k^\prime}$ is monotone, additional constraints on the coefficients are needed, specifically $\theta_{p,k^\prime+1}\geq\theta_{p,k^\prime}, \forall k', p$. 
Denote $E$ as the following $k\times(k+1)$ matrix:
\begin{align*}
    E \equiv 
     \scalebox{0.8}{$\begin{pmatrix}
        -1 & 1 & 0 & \cdots & 0 & 0 \\
        0 & -1 & 1 & \cdots & 0 & 0 \\
        \vdots & \vdots & \vdots & \ddots & \vdots & \vdots \\
        0 & 0 & 0 & \cdots & -1 & 1 
    \end{pmatrix},$}
\end{align*}
Then, the sieve policy class $\Pi_k$ is described
\begin{equation*} %\label{eqn:sieve_class}
    \Pi_k \equiv \left\{\pi_{k,\theta}: \pi_{k,\theta}(x)= \sum_{p=1}^{d_X}\sum_{k^\prime=0}^k \theta_{p,k^\prime}\varphi_{k,k^\prime}(x_p), \theta\in\mathbb{R}^{(k+1)\times d_X}, \text{ with } E\theta_p\geq 0,\forall p\right\}.
\end{equation*}
Since $\Pi_k$ lies within a finite-dimensional vector space, its VC dimension is bounded by its dimension $(k+1)d_X$. The sieve approximation rate in condition (\ref{eqn:deficiency-rate}) can be taken to be $\alpha_k = O(k^{-1})$, as established in Lemma \ref{lm:monotone-sieve-approx} in Appendix \ref{sec:proof-main}.

\subsection{Deep neural networks}

Neural networks have emerged as a viable option for approximating continuous functions. We briefly explain how to construct deep neural nets to implement policy learning. 

A function $\psi_{\text{NN}}$ on $\mathcal{X}$ implemented by a ReLU neural network can be represented as $x  \xrightarrow[]{\mathcal{L}_0} hd_1 \xrightarrow[]{\text{ReLU}} \widetilde{hd}_1 \cdots \xrightarrow[]{\mathcal{L}_{k-1}} hd_{k} \xrightarrow[]{\text{ReLU}} \widetilde{hd}_{k} \xrightarrow[]{\mathcal{L}_k} hd_{k+1} = \psi_{\text{NN}}(x)$, or more compactly expressed as $\psi_{\text{NN}} = \mathcal{L}_k \circ \text{ReLU} \circ \mathcal{L}_{k-1} \circ \text{ReLU} \circ \cdots \circ \mathcal{L}_1 \circ \text{ReLU} \circ \mathcal{L}_0$,
where $\mathcal{L}_{k'}$ is an affine transformation for each $k'=0,\cdots,k$, that is, $\mathcal{L}_{k'}(\cdot) = \Omega_{k'} \cdot + \omega_{k'}$ for some $\Omega_{k'} \in \mathbb{R}^{N_{k'+1} \times N_{k'}}$ and $\omega_{k'} \in \mathbb{R}^{N_{k'+1}}$. The ReLU activation function takes $\max\{\cdot,0\}$. In the literature, the matrix $\Omega_{k'}$ is called the weight and $\omega_{k'}$ the bias. The sieve index $k$ is the depth of the network. The integer $N_{k'}$ represents the width of the $k'$th layer, which is the number of neurons in the $k'$th layer. In particular, $N_0 = d_X$ and $N_{L+1}=1$, indicating $d_X$ inputs and a single output. A deep neural network is characterized by increasing depth while maintaining fixed width, i.e., $\Pi_k = \{\psi_{\text{NN}} \text{ with } k \text{ layers and fixed width} \}$. 

Recent work has established theoretical properties for deep neural networks. Theorem 7 of \cite{bartlett2019nearly} shows that the VC dimension of $\Pi_k$ admits a nearly tight upper bound of order $O(k^{2}\log k)$. The approximation rate of deep neural networks for smooth functions is provided in \cite{shen2021deep,shen2022optimal}. For example, by Corollary 1.3 in \cite{shen2022optimal}, for the H\"older space of continuous function of order $\gamma \in (0,1]$,\footnote{Here, fix a constant $L>0$, the H\"older space is the set of all functions $f$ satisfying $|f(x) - f(x')| \leq L \lVert x - x' \rVert_2^\gamma$.} condition (\ref{eqn:deficiency-rate}) holds with $\alpha_k \propto k^{-2\gamma/d_X}$.

\section{Empirical Study: Optimal Job Training Durations} \label{sec:applications}

In this section, we apply our proposed method to assigning individuals to job training of varying lengths, using data from the Job Training Partnership Act (JTPA) study.\footnote{The data is sourced from \cite{kitagawa2018should}, with background information extracted from the \texttt{expbif.dta} dataset, publicly available on the W.E. Upjohn Institute for Employment Research website. Observations with missing values for the included covariates have been excluded. The code for our empirical study is available on \url{https://github.com/yuefang11/continuous_policy_learning.git}.} Individuals often enroll in job training programs for varying durations. In the dataset, $22\%$ of individuals enrolled in job training for less than one month, $34.2\%$ received training for more than one month but less than three months, and $20\%$ for more than six months. We plot the frequency of the training time (measured in weeks) and its estimated density in Figure \ref{fig:t_density_freq}, demonstrating that the treatment variable is continuous.
\begin{figure}[htbp]
    \centering
    \includegraphics[width=.6\linewidth]{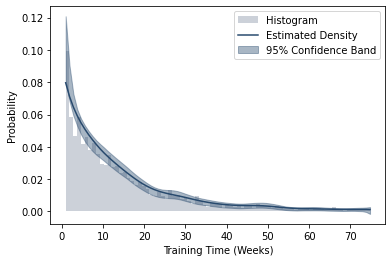}
    \caption{Frequency and estimated density of training time (weeks)}
    \label{fig:t_density_freq}
    \caption*{\footnotesize The plot displays the distribution of training times in the dataset, with density estimated using the package \texttt{lpdensity} \citep{cattaneo2022lpdensity} based on the methods developed by \cite{cattaneo2020simple,cattaneo2024local}. This highlights the rich variation in treatment, supporting its modeling as a continuous variable.}
\end{figure}

For policymakers, the challenge is not only to determine who should receive job training but also to tailor the duration of training to meet individual needs effectively. Building on prior studies, such as \cite{kitagawa2018should} and \cite{mbakop2021model}, which focus on binary eligibility for job training, our study addresses how long participants should receive training, treating the duration as a continuous variable. This shift introduces two distinctions. First, the propensity score becomes unknown, making our double-debiased method particularly effective for this setting. Second, identification requires justification, as we are no longer in an experimental context.

We argue that the unconfoundedness assumption is reasonable, as supported by \cite{flores2012estimating, hsu2022counterfactual,ColangeloLee2025} in their analysis of training duration using the Job Corps (JC) dataset. Given the close alignment between the two programs, the identification strategy developed for the JC dataset can be effectively applied to the JTPA study dataset. Both programs operated under the JTPA framework and shared similar institutional features, including open entry and exit policies, individualized training plans, and self-paced structures \citep{bloom1993national, bloom1997benefits, doolittle1993design}. This strong alignment, combined with rich pre-treatment demographic data and comparable participant characteristics, supports the validity of the unconfoundedness assumption for the JTPA dataset. Additionally, following \cite{flores2012estimating}, we restrict our analysis to individuals who have completed at least one week (40 hours) of training. This criterion ensures that the sample includes participants who have engaged with the training program to a minimal degree, enabling a more meaningful assessment of the training's effects.

Following \cite{ColangeloLee2025}, we define the continuous treatment $T$ as the total time spent in academic and vocational training (measured in weeks).\footnote{Participants are encouraged to search for employment while still in training, so the actual length of participation need not coincide with the nominal period assigned at entry. Hence, our policy in this setting is best interpreted as a recommended duration rather than a rigidly enforced requirement.} The outcome variable $Y$ is the applicants' earnings for $30$ months following the program, subtracting the training program cost, which is, on average, $\$5$ per hour \citep{bloom1997benefits}. The policy is based on three variables: years of education ($X_1$), pre-treatment earnings ($X_2$, measured by dollars per year), and working experience ($X_3$, measured by total weeks worked in the previous year). While other covariates are utilized to estimate the nuisance functions, due to legal and ethical considerations, they are only used for de-confounding but are not included in the policy.\footnote{Those covariates include gender, race,  age, location, and site of enrollment, etc.}  There are 2740 observations in our sample. The average participant is 32.39 years old, with a 40.58\% likelihood of being male, has 11.73 years of education, \$2952.80 of pre-treatment earnings, and 22.08 weeks of work experience in the previous year.

The policy spaces we consider are those defined in Section \ref{sec:monotone-sieve}. This monotonicity constraint reflects the assumption that individuals with lower levels of education, pre-treatment earnings, and work experience should receive at least as much training time as those with higher levels of these attributes. The rationale is that individuals with fewer resources or lower baseline levels in these areas may benefit more from extended training, helping them achieve outcomes comparable to those of their more advantaged peers. 

The welfare is estimated based on the double-debiased method. The estimated order of smoothness is equal to one. The propensity score $f$ is estimated using nearest neighbor kernel density estimation, and the conditional outcome function $m$ is estimated with linear regression. More complicated methods (e.g., random forests) could be adopted at the cost of increasing optimization time. For the Rademacher penalty, 100 random draws are used to simulate the Rademacher complexity for each fold, with one optimization problem solved for each draw. For the holdout penalty, 20\% of the sample is used as the testing sample. The optimization model is implemented with Gurobi 11.0 in a Python 3.11.9 environment. The factor $\gamma$ in the term $(1+\gamma)\hat{B}(h)$ is set to be 0.1. The bandwidth set $\mathcal{H}$ is the exponential sequence provided in Section \ref{sec:large}. 

Figures \ref{fig:fully_rad} and \ref{fig:fully_hold} depict the policy learning results obtained using the Rademacher and holdout penalties, respectively. For each method, \(\hat{\pi}_{k,\hat{h}_k}\) is presented for \(k = 1, \dots, 9\), where \(\hat{h}_k\) denotes the bandwidth that maximizes the penalized welfare for the corresponding \(k\). The selected policies are similar, with \((k=4, h=0.1)\) chosen under the Rademacher penalty and \((k=5, h=0.1)\) under the holdout penalty. Figure \ref{fig:slices} illustrates slices of the learned policy.

These results offer several insights. First, the learned policy yields an average training duration of approximately 11 weeks in both cases, with extended training periods—particularly those exceeding 12 weeks—rarely assigned. This observation could be attributed to the tendency of longer training durations to reduce participants' motivation to seek employment, as the training often provides subsidies or support services that may diminish the urgency of job searching.

Second, education level emerges as the primary factor influencing the assignment of training durations, whereas earnings and work experience appear to have insignificant effects. The complementary relationship between education and training may explain this distinction. Job training enhances and builds upon the foundational skills acquired through education, suggesting that individuals with lower levels of education require longer training durations to address skill gaps effectively. In contrast, earnings and work experience are typically outcomes of an individual's existing skills and reflect their current position in the labor market rather than their potential to benefit from additional training.

From an econometric perspective, we note the following observations: First, the policies learned through the Rademacher and holdout procedures are similar, yielding consistent results. Second, for a fixed sieve index \(k\), the bandwidth has a relatively minor impact on penalized welfare. Third, across the complete set of results, we observe a tendency for the selected bandwidth to increase with \(k\), aligning with our theoretical understanding of the relationship between the two tuning parameters.

\begin{figure}[htbp]
    \centering
    \includegraphics[width=\linewidth]{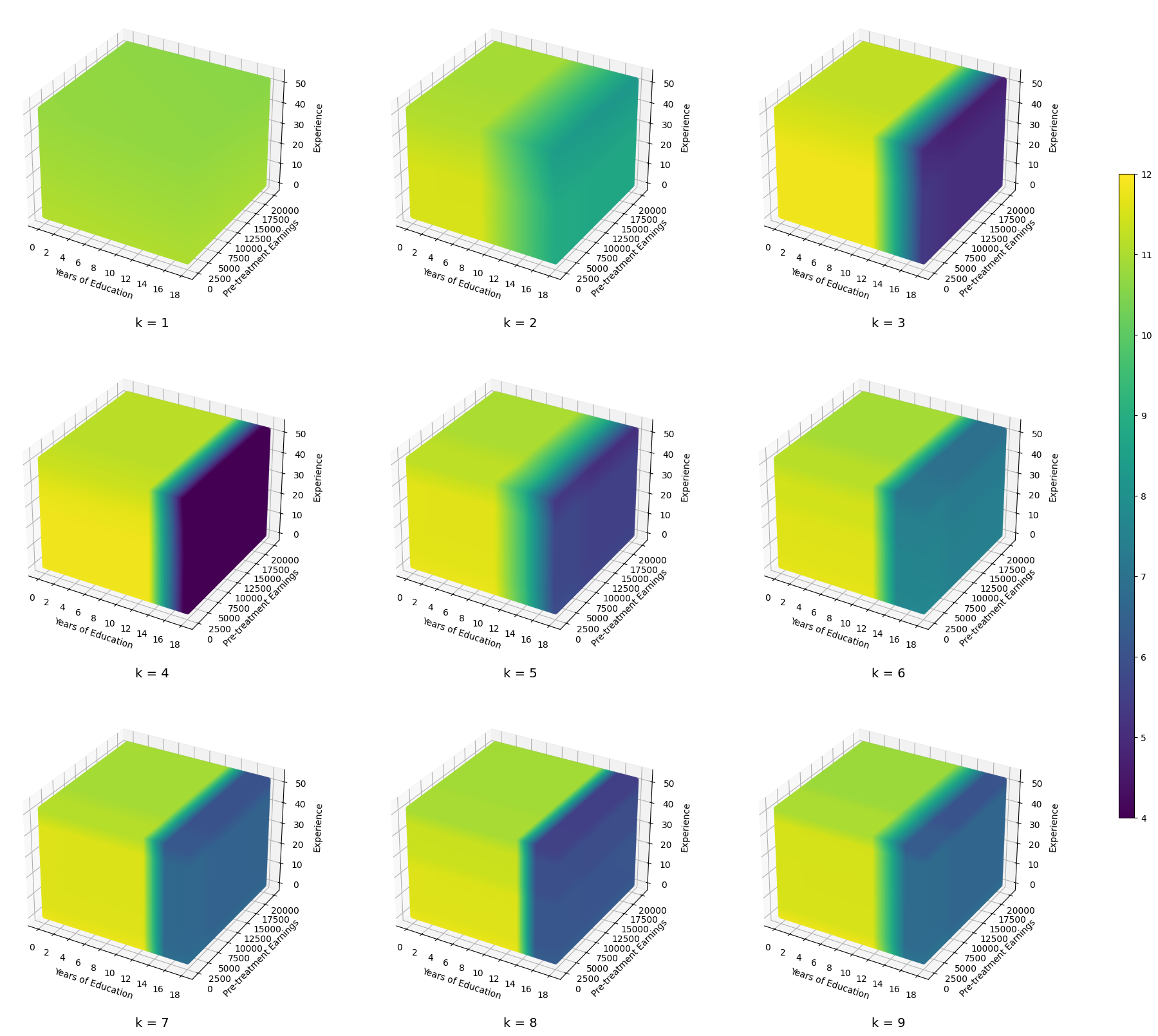}
    \caption{Optimal Training Duration (Rademacher penalty)}
    \caption*{The plots display \(\hat{\pi}_{k,\hat{h}_k}\) for \(k = 1, \dots, 9\), where \(\hat{h}_k\) represents the bandwidth that maximizes the penalized welfare for each \(k\). These plots are presented as 3-dimensional color maps, with lighter shades indicating longer training durations. Pre-treatment earnings are measured in dollars per year, working experience is measured by the total number of weeks worked in the previous year, and the training duration is measured in weeks. The policy learning procedure using the Rademacher penalty selects \(k = 4\) and \(h = 0.1\). Under this policy, the average duration of training is 10.96 weeks.}

    \label{fig:fully_rad}
\end{figure}

\begin{figure}[htbp]
    \centering
    \includegraphics[width=\linewidth]{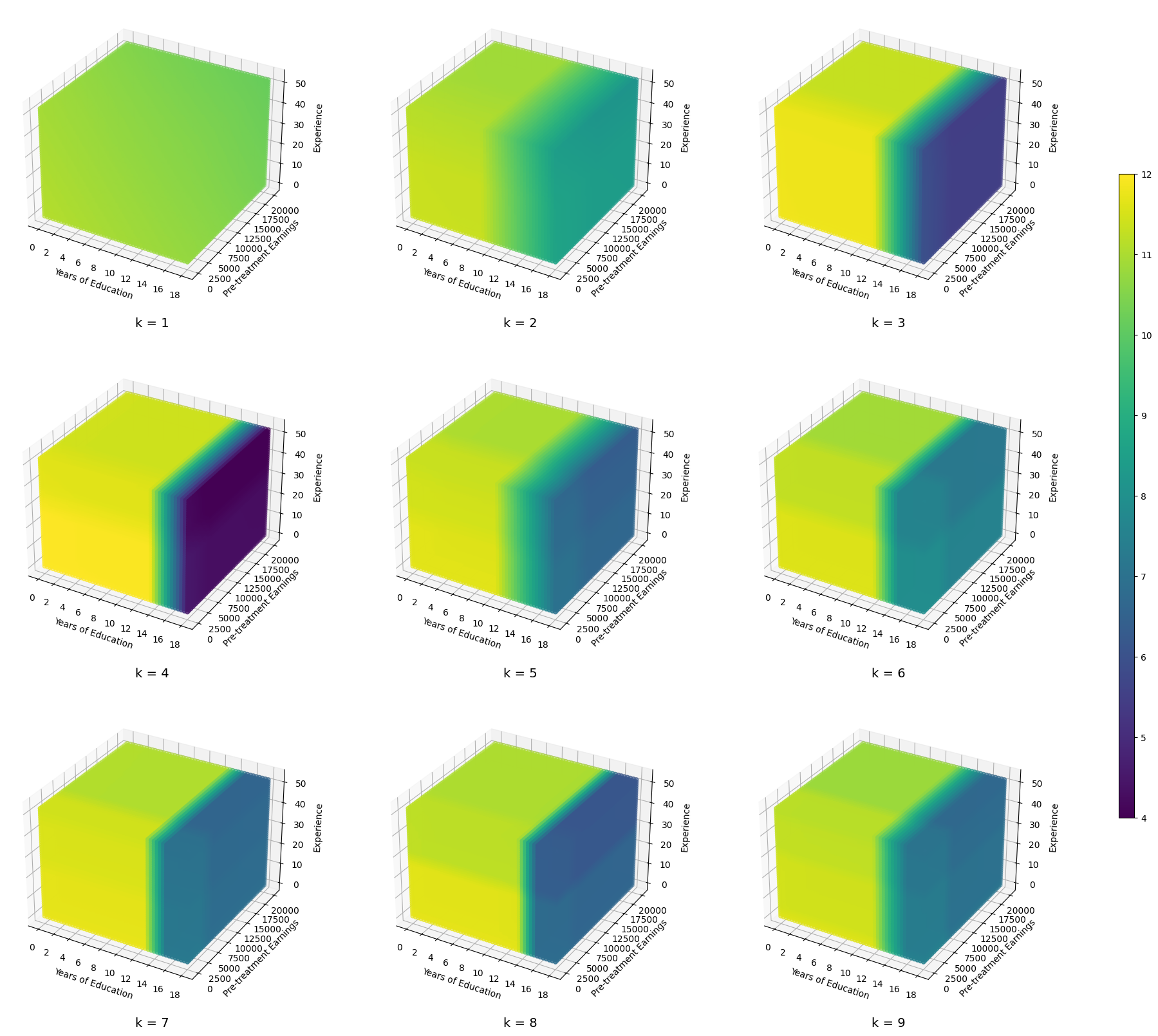}
    \caption{Optimal Training Duration (holdout penalty)}
    \caption*{The plots are analogous to those in Figure \ref{fig:fully_rad}, but are estimated using 80\% of the sample. The remaining 20\% serves as the testing sample to compute the holdout penalty. The policy learning procedure using the holdout penalty selects \(k = 5\) and \(h = 0.1\). Under this policy, the average duration of training is 10.86 weeks.}

    \label{fig:fully_hold}
\end{figure}

\begin{figure}
    \centering
    \includegraphics[width=\linewidth]{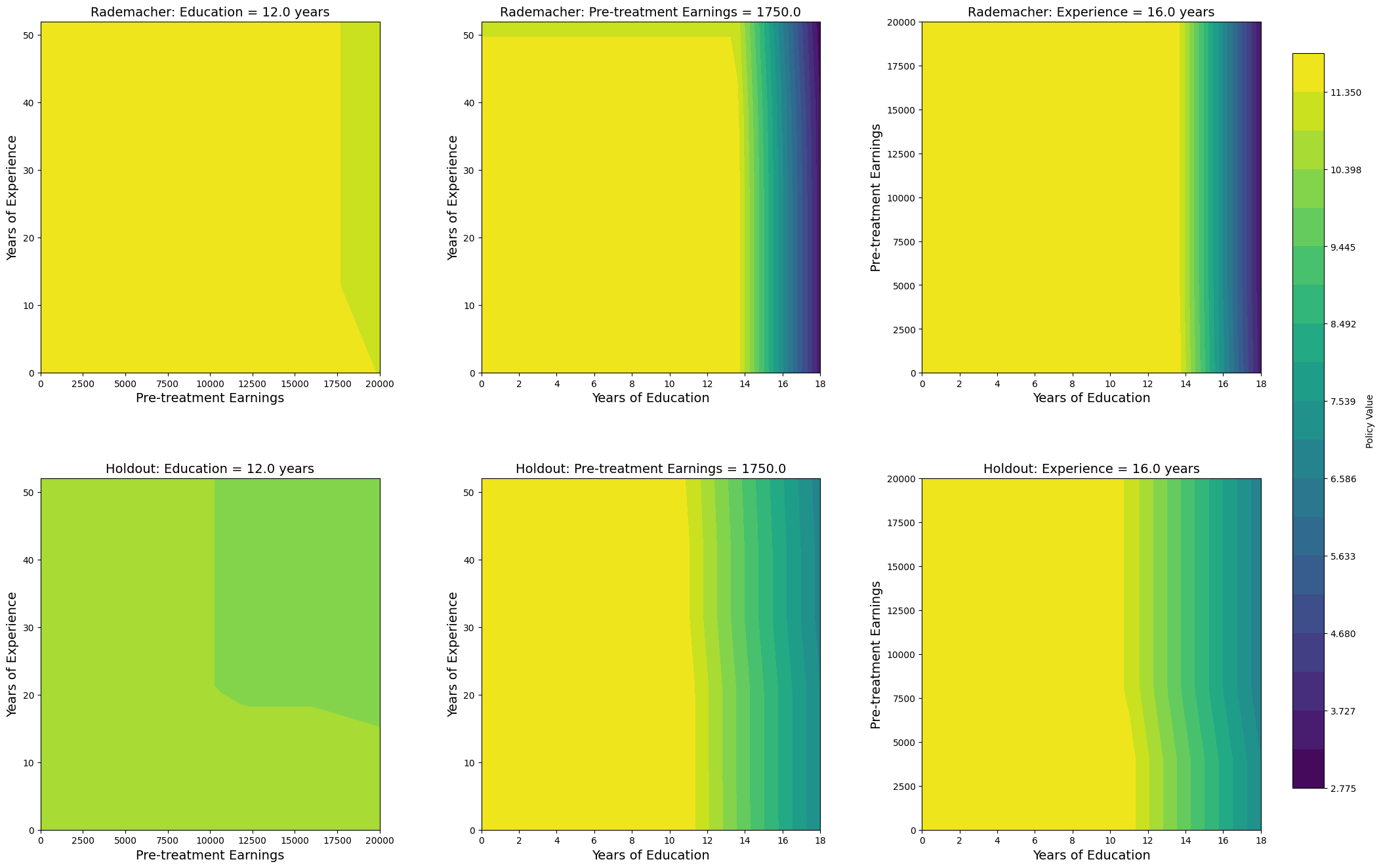}
    \caption{Slices of Optimal Policies at Medians}
    \caption*{The first row presents slices of the optimal policy from Figure \ref{fig:fully_rad}, taken at the median levels of years of education, pre-treatment earnings, and work experience, respectively. The second row presents the corresponding slices of the optimal policy from Figure \ref{fig:fully_hold}.}
    \label{fig:slices}
\end{figure}

\section{Conclusion}

This paper investigates policy learning in settings where the treatment variable is continuous. Following the framework of \cite{mbakop2021model}, we approximated the infinite-dimensional policy space using a sequence of finite-dimensional subspaces. However, in contrast to the binary treatment setting, the welfare function for continuous treatments required nonparametric estimation, even with a known propensity score. This nonparametric estimation introduced an additional tuning parameter—the bandwidth—which interacted with the dimension of the approximating space. We developed a data-automated penalization procedure for simultaneously selecting the tuning parameters. The penalty term was designed to control overfitting and account for the bias induced by the nonparametric estimation. We established oracle inequalities to demonstrate that the policy learned through this procedure effectively balanced the trade-offs between variance, the bias arising from the policy space approximation, and the kernel estimation bias. We proposed a double-debiased procedure for cases with an unknown propensity score that achieved a similar oracle inequality. We applied the proposed methodology to the JTPA dataset to determine the optimal training durations for participants based on their profiles. Compared to binary decisions regarding program participation, our approach offered more personalized recommendations for training durations. These findings suggest that decision-makers in continuous treatment settings should strongly consider using data-driven policy learning procedures.

\section*{Acknowledgment}

The authors contributed equally to this work and are listed in alphabetical order. We appreciate the Co-Editor Michael Jansson, an associate editor, and three referees whose suggestions have significantly improved this article. The National Natural Science Foundation of China (No. T2261160400, No. 72133005) supported Chunrong Ai's work, the National Natural Science Foundation of China (No. 72503208) supported Yue Fang's work, and the National Natural Science Foundation of China (No. 72403008, No. 72495123) supported Haitian Xie's work.

\bibliographystyle{chicago}
\bibliography{continuous_policy_learning.bib}

\appendix 
\begin{center}
{\large\bf APPENDIX}
\end{center}

\numberwithin{theorem}{section}
\numberwithin{lemma}{section}
\numberwithin{equation}{section}

\section{Proofs for Results in the Main Text and Appendix} \label{sec:proof-main}

\begin{lemma}\label{lm:bias-bound}
    Under Assumptions \ref{ass:unconfoundedness}, \ref{ass:infinite-kernel}, and \ref{ass:smoothness-r}, we have 
    \begin{align} \label{eqn:bound-mFT}
        |\mu^{\operatorname{FT}}(\xi)| \leq V_{\mu} |\xi|^{-(r+1)},
    \end{align}
    and the bias of $\hat{W}_{h}$ is bounded as
    \begin{align*}
        \sup_{\pi}|{W}_{h}(\pi) - W(\pi)| \leq B(h;r,V_{\mu}) \asymp h^r, 
    \end{align*}
    where the supremum is taken over the set of all measurable policies, and the expression of the bias bound is
    \begin{align*}
        B(h) = B(h;r,V_{\mu}) \equiv \frac{1}{2 \pii} \int |1- K^{\operatorname{FT}}(h \xi)|  V_{\mu} |\xi|^{-(r+1)} d\xi.
    \end{align*}
\end{lemma}

\begin{proof}[Proof of Lemma \ref{lm:bias-bound}]
    Denote $m^{\operatorname{FT}}(\xi,x) \equiv \int m(t,x)e^{\mathbf{i}t\xi} dt$ as the Fourier transform of $m$ with respect to $t$. By Lemma 1(i) in \cite{schennach2020bias} and Assumption \ref{ass:smoothness-r}(2), we have
    \begin{align*}
        |m^{\operatorname{FT}}(\xi,x)| \leq V(m(\cdot,x)) |\xi|^{-(r+1)} \leq V_{\mu} |\xi|^{-(r+1)}.
    \end{align*}
    Then we use the inverse Fourier formula to transform $W$ and $W_h$. The function $m(t,x)$ can be written as
    \begin{align*}
        m(t,x) = \frac{1}{2\pii} \int m^{\operatorname{FT}}(\xi,x) e^{-\mathbf{i}t\xi} d\xi.
    \end{align*}
    The welfare can be written as
    \begin{align*}
        W(\pi) = & \int m(\pi(x),x) f_X(x) dx = \frac{1}{2\pii} \int \int m^{\operatorname{FT}}(\xi,x) e^{-\mathbf{i}\pi(x)\xi} d\xi f_X(x) dx.
    \end{align*}
    For the kernel term, a standard change of variables gives
    \begin{align*}
        \frac{1}{h}K \left( \frac{\pi(X) - T}{h} \right) = & \frac{1}{2\pii h} \int K^{\operatorname{FT}}(\xi) \exp(-\mathbf{i}(\pi(X)-T)\xi/h) d\xi \\
        = & \frac{1}{2\pii} \int K^{\operatorname{FT}}(h\xi) e^{-\mathbf{i}\pi(X)\xi} e^{\mathbf{i}T\xi} d\xi.
    \end{align*}
    Using the above expression, we can write the expectation of the welfare estimator as
    \begin{align*}
        W_{h}(\pi) = & \mathbb{E}\left[\frac{1}{h} K \left( \frac{\pi(X) - T}{h} \right) \frac{m(T,X)}{f(T|X)}  \right] \\
        = & \frac{1}{2\pii} \mathbb{E}\left[ \int \int K^{\operatorname{FT}}(h\xi) e^{-\mathbf{i}\pi(X)\xi} e^{\mathbf{i}t\xi} m(t,X) dt d\xi \right] \\
        = & \frac{1}{2\pii} \mathbb{E}\bigg[ \int K^{\operatorname{FT}}(h\xi) e^{-\mathbf{i}\pi(X)\xi} \underbrace{\int e^{\mathbf{i}t\xi} m(t,X) dt}_{= m^{\operatorname{FT}}(\xi,X)} d\xi \bigg] \\
        = & \frac{1}{2\pii} \mathbb{E}\left[ \int K^{\operatorname{FT}}(h\xi) e^{-\mathbf{i}\pi(X)\xi} m^{\operatorname{FT}}(\xi,X) d\xi \right],
    \end{align*}
    where the third lines follows from switching the order of integration by Fubini theorem.
    Therefore, the bias is characterized as
    \begin{align*}
        |W_{h}(\pi) - W(\pi)| = & \frac{1}{2\pii} \left| \int \int (1-K^{\operatorname{FT}}(h\xi)) m^{\operatorname{FT}}(\xi,x) e^{-\mathbf{i}\pi(x)\xi} d\xi f_X(x) dx \right| \\
        \leq & \frac{1}{2\pii} \int \int |1-K^{\operatorname{FT}}(h\xi)) ||m^{\operatorname{FT}}(\xi,x) | d\xi f_X(x) dx \\
        \leq & \int B(h;r,V_{\mu}) f_X(x) dx = B(h;r,V_{\mu}).
    \end{align*}
    The fact that $B(h;r,V_{\mu}) \asymp h^r$ is proved by Lemma 2(ii) in \cite{schennach2020bias}.
\end{proof}

\begin{lemma} \label{lm:r-hat-consistency}
    Let Assumptions \ref{ass:bounded} and \ref{ass:bias-est} hold. %Assume that $\hat{\mu}$ and $\hat{f}_T$ be constructed using cross-fitting with a finite number of folds and satisfy the following conditions:
    %\begin{enumerate}[label = (\arabic*)] 
        %\item $\hat{f}_T$ bounded away from zero, 
        %\item $\mathbb{E}[\int (\hat{\mu}(t) - \mu(t))^2 f_T(t) dt], \mathbb{E}[\int (\hat{f}_T(t) - f_T(t))^2 f_T(t) dt] \leq n^{-\epsilon}$ for some $\epsilon>0$,
        %\item $\lVert \hat{\mu} - \mu \rVert_\infty \lVert \hat{f}_T - f_T \rVert_\infty = o_{a.s.}(n^{-1/2})$.
    %\end{enumerate}
    Then we have $\mathbb{P}(\hat{r}=r) \rightarrow 1$ and $\hat{V}_{\mu} = V_{\mu} + o_p(1)$. 
\end{lemma}

\begin{proof}[Proof of Lemma \ref{lm:r-hat-consistency}]
    For simplicity in the exposition, we assume that $\hat{\mu}$ and $\hat{f}_T$ are independent of the sample, essentially achieved by the cross-fitting method.
    The proof is based on Theorem 1 of \cite{schennach2020bias}. The only difference is that we do not observe $Y'$, and therefore, the estimator of $\mu^{\operatorname{FT}}$ involves nuisances estimators. The goal is to show that the uniform convergence result in Lemma A.4 of \cite{schennach2020bias} still holds for our estimator $\hat{\mu}^{\operatorname{FT}}$. Consider the following infeasible estimator
    \begin{align*}
        \tilde{\mu}^{\operatorname{FT}}(\xi) \equiv \frac{1}{n} \sum_{i=1}^n \frac{Y_i - \mu(T_i)}{f_T(T_i)} e^{iT_i\xi} + \mu^{\operatorname{FT}}(\xi),
    \end{align*}
    which satisfies Lemma A.4 of \cite{schennach2020bias}.
    We want to show that the double-debiased estimator, $\hat{\mu}^{\operatorname{FT}}$, is sufficiently close to $\tilde{\mu}^{\operatorname{FT}}$ uniformly. The difference between $\hat{\mu}^{\operatorname{FT}}$ and $\tilde{\mu}^{\operatorname{FT}}$ is decomposed into the following three terms:
    \begin{align*}
        & \int (\hat{\mu}(t) - \mu(t))(\hat{f}_T(t) - f_T(t))\frac{e^{\mathbf{i}t\xi}}{\hat{f}_T(t)}  dt, \\
        & \frac{1}{n}\sum_{i=1}^n \left( \frac{1}{\hat{f}_T(T_i)} - \frac{1}{f_T(T_i)} \right) (Y_i - \mu(T_i)) e^{iT_i\xi}, \\
        & \frac{1}{n}\sum_{i=1}^n \frac{\mu(T_i) - \hat{\mu}(T_i)}{\hat{f}_T(T_i)}e^{iT_i\xi} - \int (\hat{\mu}(t) - \mu(t))\frac{e^{\mathbf{i}t\xi}}{\hat{f}_T(t)} f_T(t) dt.
    \end{align*}
    The first term is $o_{a.s.}(n^{-1/4})$ because $T$ has bounded support, $\hat{f}_T$ is bounded away from zero, and $\lVert \hat{\mu} - \mu \rVert_\infty \lVert \hat{f}_T - f_T \rVert_\infty = o_{a.s.}(n^{-1/2})$. For the second term, we use Kolmogorov's three-series theorem. Define $Z_{j} \equiv \left( \frac{1}{\hat{f}_T(T_j)} - \frac{1}{f_T(T_j)} \right) (Y_j - \mu(T_j)) e^{iT_j\xi}$which has a mean zero. Then, by the boundedness of $Y$, $1/f_T$, and $1/\hat{f}_{T}$, we have 
    \begin{align*}
        \sum_{n=1}^\infty \mathbb{E}|Z_n|^2/n \lesssim \sum_{n=1}^\infty  \mathbb{E}\int (\hat{f}_T(t) - f_T(t))^2 f_T(t)^2/n \leq \sum_{n=1}^\infty n^{-1-\epsilon} < \infty.
    \end{align*}
    According to Kolmogorov's three-series theorem,  $\sum_{n=1}^\infty Z_n/\sqrt{n}$ converges almost surely. Then, according to Kronecker's Lemma \citep[e.g., Lemma 7.4.1 in][]{resnick2005}, we have $\sum_{i=1}^n Z_i/\sqrt{n} = o_{a.s.}(1)$, which implies that the second term in decomposition is of order $o_{a.s.}(n^{-1/2})$. Following the same procedure, we can derive that the third term in the decomposition is also of order $o_{a.s.}(n^{-1/2})$. I noticed that the above convergence is all uniform  $\xi$ because $|e^{\mathbf{i}t\xi}| \leq 1$. Therefore, we have $\hat{\mu}^{\operatorname{FT}}(\xi)-\tilde{\mu}^{\operatorname{FT}}(\xi) = o_{a.s.}(n^{-1/2})$ uniformly in $\xi$. Then rate results in Lemma A.4 of \cite{schennach2020bias} would also apply to $\hat{\mu}^{\operatorname{FT}}(\xi)$. The remaining parts of the proof are the same as that of Theorem 1 in \cite{schennach2020bias}, which establishes the consistency of $\hat{r}$ and $\hat{V}_{\mu}$.
\end{proof}

\begin{proof}[Proof of Theorem \ref{thm:fully}]
In the proof, we will use \(C, C_1, C_2,\) and so forth to represent constants that are independent of \((n, k, h)\). These constants may vary at different points in the proof.
For the sake of simplicity, we will assume that there exists $\pi_k^* \in \Pi_k, k=1,\cdots,\infty,$ such that $W^*(\Pi_k) = W(\pi_k^*)$. If not, one can also form a sequence of policies with welfare approaching $W(\pi_k^*)$. 
Decompose the welfare regret into the following two parts:
\begin{align*}
    W(\pi_\infty^*) - W(\hat{\pi}) = \left(W(\pi_\infty^*) - \hat{Q}_{\hat{h},\hat{k}}\right) +\left(\hat{Q}_{\hat{h},\hat{k}} - W(\hat{\pi})\right).
\end{align*}
The first term is the leading term, and the second is the remainder. We first deal with the leading term. Fix $h \in \mathcal{H}$ and $k$ such that $\operatorname{VC}(\Pi_k) < nh^2$, we have
\begin{align*}
    W(\pi_\infty^*) - \hat{Q}_{\hat{h},\hat{k}} = \underbrace{W(\pi_\infty^*) - W(\pi_k^*)}_{\text{welfare deficiency}} + W(\pi_k^*) - \hat{Q}_{\hat{h},\hat{k}}.
\end{align*}
The first term is welfare deficiency. The second term is bounded by 
\begin{align} \label{eqn:W-Q-bound}
    W(\pi^*_k) - \hat{Q}_{\hat{h},\hat{k}} & \leq W(\pi^*_k) - \hat{Q}_{h,k} \nonumber \\
    & = W(\pi^*_k) - \hat{W}_{h}(\hat{\pi}_{h,k}) + \hat{R}_{h,k} + B(h) + \tau(h,k,n) \text{ (definition of $\hat{Q}_{h,k}$)} \nonumber \\
    & \leq W(\pi^*_k) - \hat{W}_{h}(\pi^*_k) + \hat{R}_{h,k} + B(h) + \tau(h,k,n) \text{ (definition of $\hat{\pi}_{h,k}$)} \nonumber \\
    & \leq \sup_{\pi \in \Pi_k} |\hat{W}_{h}(\pi) - W(\pi)| + \hat{R}_{h,k} + B(h) + \tau(h,k,n) \nonumber \\
    & \leq \Delta_{h,k} + \hat{R}_{h,k} + 2B(h) + \tau(h,k,n).
\end{align}
where $\Delta_{h,k} \equiv \sup_{\pi \in \Pi_k} |\hat{W}_h(\pi) - W_h(\pi)|$.
On the right-hand side of the above inequality, the random terms are $\Delta_{h,k}$ and $\hat{R}_{h,k}$, both of which can be bounded using similar arguments: we first bound their expectations using Lemma \ref{lm:kitagawa}, and then apply Talagrand’s inequality to control their deviations from the mean. We begin with $\hat{R}_{h,k}$. Define $\bar{R}_{h,k} \equiv \mathbb{E}[\hat{R}_{h,k}]$ as the (expected) Rademacher complexity, which is the expected supremum of the empirical process indexed by the following class of functions
    \begin{align*}
        \left\{(Y,T,X,\text{Rad}) \mapsto \frac{2}{h} \text{Rad} K \left( \frac{T - \pi(X)}{h} \right) \frac{Y}{f(T|X)}, \pi \in \Pi_k \right\}.
    \end{align*}
    Notice that the functions in this class have zero means due to the independent Rademacher variables. According to Lemma \ref{lm:kernel-VC-dim}, the VC dimension of this function class is bounded by $2\operatorname{VC}(\Pi_k)$.
    Also, this class of functions admits a uniform bound $\frac{2\bar{\kappa}M}{h \underline{f}}$ and a second-moment bound $\frac{4M^2\kappa_2}{h \underline{f}}$ (Lemma \ref{lm:uniform-variance-bound}).
    Plugging the VC dimension, the uniform bound, and the second-moment bound into (\ref{eqn:EDelta-bound2}) of Lemma \ref{lm:kitagawa}, we obtain that $\bar{R}_{h,k}$ is bounded as
    \begin{align*}
        \bar{R}_{h,k} \leq & 2 \frac{2\bar{\kappa}M}{h \underline{f}} c^2 \frac{2\operatorname{VC}(\Pi_k)}{n} + c \sqrt{\frac{4M^2 \kappa_2}{h\underline{f}}\frac{2\operatorname{VC}(\Pi_k)}{n}} \\
        = & 8c^2 \frac{\bar{\kappa}M}{\underline{f}} \frac{\operatorname{VC}(\Pi_k)}{nh} + 2\sqrt{2}c M \sqrt{\frac{\kappa_2}{\underline{f}}} \sqrt{\frac{\operatorname{VC}(\Pi_k)}{nh}},
    \end{align*}
    where $c$ is the universal constant given by Lemma A.4 in \cite{kitagawa2018should}, which can be computed explicitly. Since the second term dominates the first term, $\bar{R}_{h,k}$ is bounded by $(C_v + o(1)) \sqrt{\frac{\operatorname{VC}(\Pi_k)}{nh}}$, where we redefine $c$ as $2\sqrt{2}c$. The probability of deviation from the mean can be bounded by using the Talagrand inequality:
\begin{align*}
    \mathbb{P}(\hat{R}_{h,k} - \mathbb{E}[\hat{R}_{h,k}] > \alpha) & \leq C \exp\left( -\frac{n \alpha^2}{\frac{C_1}{h} + \frac{C_2}{h} \sqrt{\frac{\operatorname{VC}(\Pi_k)}{nh^2}} + \frac{C_3 \alpha}{h}} \right) \\
    & \leq C \exp\left( -\frac{nh_{min} \alpha^2}{C + C \alpha} \right),
\end{align*}
where the second line follows from the construction that $\operatorname{VC}(\Pi_k) \leq nh^2$ and $h \geq (\log n)^2/n$. This tail bound, together with Lemma \ref{lm:tail-to-Op-rate}, shows that $\hat{R}_{h,k} - \mathbb{E}[\hat{R}_{h,k}] = O_p((nh_{min})^{-1/2})$ uniformly over $h$ and $k$. Similarly, for the term $\Delta_{h,k}$ on the right-hand side of (\ref{eqn:W-Q-bound}), we can first bound its mean by $\bar{R}_{h,k}$ by using he standard symmetrization argument \citep[see, for example, Lemma 2.3.1 in][]{wellner1996}. The deviation $\Delta_{h,k} - \mathbb{E}[\Delta_{h,k}]$ can again be bounded by applying Talagrand’s inequality.
Plugging the bounds on $\Delta_{h,k}$ and $\bar{R}_{h,k}$ into (\ref{eqn:W-Q-bound}), we obtain that
\begin{align*}
    W(\pi_k^*)  - \hat{Q}_{\hat{h},\hat{k}} & \mathbb{E}[\Delta_{h,k}] + \mathbb{E}[\hat{R}_{h,k}] + 2B(h) + \tau(h,k,n) + (\Delta_{h,k} - \mathbb{E}[\Delta_{h,k}] + \hat{R}_{h,k}-\mathbb{E}[\hat{R}_{h,k}]) \\
    & \leq 2 (C_v + o(1)) \sqrt{\frac{\operatorname{VC}(\Pi_k)}{nh}} + 2 B(h) + \tau(h,k,n) + O_p\left( (nh_{min})^{-1/2} \right).
\end{align*}
Since the choice of $(h,k)$ is arbitrary and the $O_p$-terms are uniform in $(h,k)$, we obtain that 
\begin{align*}
    W(\pi_\infty^*)  - \hat{Q}_{\hat{h},\hat{k}} & \leq \inf_{h \in \mathcal{H},k:\operatorname{VC}(\Pi_k) \leq nh^2} \left(  W(\pi_\infty^*) - W(\pi_k^*) + 2 (C_v + o(1)) \sqrt{\frac{\operatorname{VC}(\Pi_k)}{nh}} + 2 B(h) + \tau(h,k,n) \right) \\
    & + O_p\left( (nh_{min})^{-1/2} \right).
\end{align*}
Then, we deal with the remainder term $\hat{Q}_{\hat{h},\hat{k}} - W(\hat{\pi}_{\hat{h},\hat{k}})$. We have
\begin{align*}
    & \hat{Q}_{\hat{h},\hat{k}} - W(\hat{\pi}_{\hat{h},\hat{k}}) \\
    = & \hat{W}_{\hat{h}}(\hat{\pi}_{\hat{h},\hat{k}}) - W_{\hat{h}}(\hat{\pi}_{\hat{h},\hat{k}}) -\hat{R}_{\hat{h},\hat{k}}  - \tau(\hat{h},\hat{k},n) + \underbrace{W_{\hat{h}}(\hat{\pi}_{\hat{h},\hat{k}}) - W(\hat{\pi}_{\hat{h},\hat{k}}) -  B(\hat{h})}_{\leq 0 \text{ (Lemma \ref{lm:bias-bound})}} \\
    \leq & \Delta_{\hat{h},\hat{k}} -\hat{R}_{\hat{h},\hat{k}}  - \tau(\hat{h},\hat{k},n),
\end{align*}
where recall that $\Delta_{h,k} \equiv \sup_{\pi \in \Pi_k} |\hat{W}_h(\pi) - W_h(\pi)|$.
The right tail bound for the above term can be derived by using the union bound as follows: 
\begin{align} \label{eqn:union-bound}
& \mathbb{P}(\Delta_{\hat{h},\hat{k}} -\hat{R}_{\hat{h},\hat{k}} - \tau(\hat{h},\hat{k},n) > \alpha) \nonumber \\
    \leq & \sum_{k, h} \mathbb{P}(\Delta_{h,k} - \bar{\Delta}_{h,k} > (\alpha + \tau(h,k,n))/2) + \sum_{k, h} \mathbb{P}(\hat{R}_{h,k} - \bar{R}_{h,k} > (\alpha + \tau(h,k,n))/2),
\end{align}
where $\bar{\Delta}_{h,k} \equiv \mathbb{E}[\Delta_{h,k}]$ and $\bar{R}_{h,k} \equiv \mathbb{E}[R_{h,k}]$. We only need to analyze the first probability as the second one is similar. Using Talagrand inequality again, we obtain that 
\begin{align} \label{eqn:tail-bound}
    \mathbb{P}(\Delta_{h,k} - \bar{\Delta}_{h,k} > (\alpha + \tau(h,k,n))/2) & \leq C \exp\left( -\frac{n (\alpha + \tau(h,k,n))^2}{\frac{C_1}{h} + \frac{C_2}{h} \sqrt{\frac{\operatorname{VC}(\Pi_k)}{nh^2}} + \frac{C_3 (\alpha+\tau(h,k,n))}{h}} \right) \nonumber \\
    & \leq C \exp\left( -\frac{nh (\alpha + \tau(h,k,n))^2}{C + C(\alpha+\tau(h,k,n))} \right),
\end{align}
where the last line follows from the restrictions that $\operatorname{VC}(\Pi_k) \leq nh^2$.
We study two cases. First, in the case of $C > C(\alpha+\tau)$, the tail bound becomes
\begin{align*}
    \exp\left( -nh (\alpha + \tau(h,k,n))^2/C \right) & \leq \exp(-nh\alpha^2/C) \exp(-nh\tau(h,k,n)^2/C) \\
    & \leq \exp(-nh_{min}\alpha^2/C) \exp(-nh\tau(h,k,n)^2/C).
\end{align*}
Therefore, assuming that $\sum_{k,h} \exp(-nh\tau(h,k,n)^2/C)$ is finite and does not grow with $n$, we have 
\begin{align*}
    \sum_{k,h} \exp\left( -nh (\alpha + \tau(h,k,n))^2/C \right) \leq C \exp(-nh_{min}\alpha^2/C).
\end{align*}
In view of Lemma \ref{lm:tail-to-Op-rate}(i), this term is $O_p((nh_{min})^{-1/2})$.
In the case where $C \leq C(\alpha+\tau)$, we have
\begin{align*}
    \exp\left( -\frac{nh (\alpha + \tau(h,k,n))^2}{C(\alpha+\tau)} \right) & = \exp\left( -nh (\alpha + \tau(h,k,n))/C \right) \\
    & \leq \exp(-nh_{min}\alpha/C) \exp(-nh\tau(h,k,n)^2/C),
\end{align*}
assuming that $\tau \in (0,1)$. Therefore, this term is an exponential tail given the same condition. In view of Lemma \ref{lm:tail-to-Op-rate}(ii), this term is $O_p((nh_{min})^{-1})$. To summarize, we have shown that the positive part of $\hat{Q}_{\hat{h},\hat{k}} - W(\hat{\pi}_{\hat{h},\hat{k}})$
is of order $O_p((nh_{min})^{-1/2})$. Combining the leading term $W(\pi_\infty^*)  - \hat{Q}_{\hat{h},\hat{k}}$ and the remainder term $\hat{Q}_{\hat{h},\hat{k}} - W(\hat{\pi}_{\hat{h},\hat{k}})$, we have shown that
\begin{align*}
    W(\pi_\infty^*) - W(\hat{\pi}_{\hat{h},\hat{k}}) & \leq \inf_{h \in \mathcal{H},k:\operatorname{VC}(\Pi_k) \leq nh^2} \left(  W(\pi_\infty^*) - W(\pi_k^*) + 2 (C_v + o(1)) \sqrt{\frac{\operatorname{VC}(\Pi_k)}{nh}} + 2B(h) + \tau(h,k,n) \right) \\
    & + O_p\left( (nh_{min})^{-1/2} \right),
\end{align*}
under the condition that $\sum_{k,h} \exp(-nh\tau(h,k,n)^2/C)$ is finite and does not grow with $n$. One sufficient condition for the remainder $O_p\left( (nh_{min})^{-1/2} \right)$ to be smaller than the leading term is to take $h_{min} \geq n^{-1/(2r+1)}$. This proves part (1) of the theorem. To verify the feasibility of the choices of $\mathcal{H}$ and $\tau$ given below Theorem \ref{thm:fully}, notice that in the case of geometric $\mathcal{H}$, we have
\begin{align*}
    \sum_{h,k} \exp(-nh\tau(h,k,n)^2) \leq \sum_{j=1}^\infty \exp(-2 \log j) \sum_{k=1}^\infty \exp(-2 \log k).
\end{align*}
The case of exponential $\mathcal{H}$ is similar.

For part (2), the bias bound is estimated using $\hat{r}$ and $\hat{V}_\mu$, which are consistent given Lemma \ref{lm:r-hat-consistency}. The smallest bandwidth in the grid is $h_{min} = n^{-1/(2\hat{r}+1)}$. To differentiate the two procedures, we denote $\hat{\pi}^B$ for part (1) and $\hat{\pi}^{(1+\gamma)\hat{B}}$ for part (2).
The analysis is conducted conditionally on the event $\mathscr{E} \equiv \{\hat{r} = r, V_\mu < (1+\gamma)\hat{V}_\mu < (1+2\gamma)V_\mu \}$. By Lemma \ref{lm:r-hat-consistency}, $\mathbb{P}(\mathscr{E}) \rightarrow 1$. Also, we have $(1+\gamma)\hat{B}(h)\mathbf{1}_{\mathscr{E}} \leq (1+2\gamma)B(h)$, and $(W_{\hat{h}}(\hat{\pi}_{\hat{h},\hat{k}}) - W(\hat{\pi}_{\hat{h},\hat{k}}) -  (1+\gamma)\hat{B}(\hat{h}))\mathbf{1}_{\mathscr{E}} \leq 0$.
Decompose the welfare regret based on whether $\mathscr{E}$ holds:
\begin{align*}
    W(\pi_\infty^*) - W(\hat{\pi}^{(1+\gamma)\hat{B}}) = (W(\pi_\infty^*) - W(\hat{\pi}^{(1+\gamma)\hat{B}})) \mathbf{1}_{\mathscr{E}} + (W(\pi_\infty^*) - W(\hat{\pi}^{(1+\gamma)\hat{B}}))\mathbf{1}_{\mathscr{E}^c}.
\end{align*}
For the first term $(W(\pi_\infty^*) - W(\hat{\pi}^{(1+\gamma)\hat{B}})) \mathbf{1}_{\mathscr{E}}$, we can follow the proof for part (1) and show that it is bounded by
\begin{align*}
    & (W(\pi_\infty^*) - W(\hat{\pi}^{(1+\gamma)\hat{B}}))\mathbf{1}_{\mathscr{E}}\\
    \leq & \inf_{h \in \mathcal{H},k:\operatorname{VC}(\Pi_k) \leq nh^2} \left(  W(\pi_\infty^*) - W(\pi_k^*) + 2 (C_v + o(1)) \sqrt{\frac{\operatorname{VC}(\Pi_k)}{nh}} + (2+2\gamma)B(h) + \tau(h,k,n) \right).
\end{align*}
To show that the second term
$$
\bigl(W(\pi_\infty^*)-W(\hat\pi^{(1+\gamma)\hat B})\bigr)\mathbf 1_{\mathscr E^c}
$$
is $O_p\!\bigl(n^{-r/(2r+1)}\bigr)$, we in fact establish a stronger result:\footnote{This result corresponds to an interesting feature of convergence in probability: the indicator of an event with vanishing probability converges to zero at an arbitrarily fast rate, regardless of how slowly the probability of the event itself decays. This behavior is specific to convergence in probability and does not extend to stronger notions such as $L_p$ convergence.
} for every $\delta>0$,
$$
\bigl(W(\pi_\infty^*)-W(\hat\pi^{(1+\gamma)\hat B})\bigr)\mathbf 1_{\mathscr E^c}=o_p\!\bigl(n^{-\delta}\bigr).
$$
Equivalently, $n^{\delta}\bigl(W(\pi_\infty^*)-W(\hat\pi^{(1+\gamma)\hat B})\bigr)\mathbf 1_{\mathscr E^c}$ converges to zero in probability. For any $\varepsilon>0$,
$$
n^{\delta}\bigl|W(\pi_\infty^*)-W(\hat\pi^{(1+\gamma)\hat B})\bigr|\mathbf 1_{\mathscr E^c}>\varepsilon
\;\;\Longrightarrow\;\;
\mathscr E^c\ \text{occurs},
$$
and therefore
$$
\mathbb{P}\Bigl(n^{\delta}\bigl|W(\pi_\infty^*)-W(\hat\pi^{(1+\gamma)\hat B})\bigr|\mathbf 1_{\mathscr E^c}>\varepsilon\Bigr)
\le
\mathbb{P}(\mathscr E^c) \rightarrow 0,
$$
which confirms that
$$
\bigl|W(\pi_\infty^*)-W(\hat\pi^{(1+\gamma)\hat B})\bigr|\mathbf 1_{\mathscr E^c}=o_p\!\bigl(n^{-\delta}\bigr),\qquad \forall\delta>0.
$$
This proves part (2) of the theorem. Note that the above $O_p$ bound for $\mathbf{1}_{\mathscr E^c}$ cannot be upgraded to an $L_p$ (expected-regret) bound, because that would require a rate for $\hat r$, which is typically very slow.

\end{proof}

\begin{proof}[Proof of Corollary \ref{cor:holdout}]
    In the holdout procedure, redefine the term $\hat{R}_{h,k}$ as $\hat{W}^E_{h}(\hat{\pi}^E_{k}) - \hat{W}^T_{h}(\hat{\pi}^E_{k})$. Then, the leading term in the welfare regret can be derived in the same way as in Theorem \ref{thm:fully}. The only difference lies in the derivation of the remainder term:
    \begin{align*}
    & \hat{Q}_{\hat{h},\hat{k}}^{\textit{hold}} - W(\hat{\pi}_{\hat{h},\hat{k}}) \\
    = & \hat{W}^E_{\hat{h}}(\hat{\pi}^E_{\hat{h},\hat{k}}) - W_{\hat{h}}(\hat{\pi}^E_{\hat{h},\hat{k}}) -\hat{R}_{\hat{h},\hat{k}}  - \tau(\hat{h},\hat{k},n) + \underbrace{W_{\hat{h}}(\hat{\pi}^E_{\hat{h},\hat{k}}) - W(\hat{\pi}^E_{\hat{h},\hat{k}}) -  B(\hat{h})}_{\leq 0 \text{ (Lemma \ref{lm:bias-bound})}} \\
    \leq & \hat{W}^T_{\hat{h}}(\hat{\pi}^E_{\hat{h},\hat{k}}) - W_{\hat{h}}(\hat{\pi}^E_{\hat{h},\hat{k}}) - \tau(\hat{h},\hat{k},n).
\end{align*}
Due to the holdout structure, the tail probability can now be bounded by after conditioning on the estimating sample:
\begin{align*}
    & \mathbb{P}\left( \hat{W}^T_{\hat{h}}(\hat{\pi}^E_{\hat{h},\hat{k}}) - W_{\hat{h}}(\hat{\pi}^E_{\hat{h},\hat{k}}) - \tau(\hat{h},\hat{k},n) > \alpha \right) \\
    \leq & \sum_{h,k} \mathbb{E}\left[ \mathbb{P}\left( \hat{W}^T_{h}(\hat{\pi}^E_{h,k}) - W_{h}(\hat{\pi}^E_{h,k}) > \alpha + \tau(h,k,n) | S_{n_E}\right) \right],
\end{align*}
The probability $\mathbb{P}\left( \hat{W}^T_{h}(\hat{\pi}^E_{h,k}) - W_{h}(\hat{\pi}^E_{h,k}) > \alpha + \tau(h,k,n) | S_{n_E}\right)$ can be bounded using Bernstein inequality \citep[e.g., Proposition 2.14 in][]{wainwright_2019}. Then we can proceed as in the proof of Theorem \ref{thm:fully}.
\end{proof}

\begin{proof}[Proof of Theorem \ref{thm:dr}]
    For simplicity, we consider the case where \( B(h) \) is known, as handling the estimated case follows the same approach as in the proof of Theorem \ref{thm:fully}(2) by conditioning on the event \( \mathscr{E} \). Similarly, we can treat the two events in Assumption \ref{ass:rates-nonparametric-estimates}(ii)–(iii) as occurring with probability one.

    The structure of the proof follows that of Theorem \ref{thm:fully}. The extra work is to bound the difference between $\hat{R}^{\operatorname{DD},\ell}_{h,k}$ with the infeasible Rademacher complexity $\tilde{R}^{\operatorname{DD},\ell}_{h,k}$ constructed using the true nuisance parameters
    \begin{align*}
        \tilde{R}^{\operatorname{DD},\ell}_{h,k} \equiv \mathbb{E} \left[ \sup_{\pi \in \Pi_k} \frac{L}{n} \sum_{i \in I_\ell} 2\text{Rad}_i \cdot \Gamma_h(Y_i,T_i,X_i;\pi;g,m) \Bigg| S_\ell \right], \tilde{R}^{\operatorname{DD}}_{h,k} \equiv \frac{1}{L} \sum_{\ell=1}^L \tilde{R}^{\operatorname{DD},\ell}_{h,k}.
    \end{align*}
    The difference between $W(\pi^*_k)$ and $\hat{Q}^{\operatorname{DD}}_{h,k}$ is bounded as
    \begin{align*}
W(\pi^*_k) - \hat{Q}^{\operatorname{DD}}_{h,k} & = W(\pi^*_k) - \hat{W}^{\operatorname{DD}}_{h}(\hat{\pi}_{h,k}) + \hat{R}^{\operatorname{DD}}_{h,k} + B(h) + \tau(h,k,n) \\
    & \leq W(\pi^*_k) - \hat{W}^{\operatorname{DD}}_{h}(\pi^*_k) + \hat{R}^{\operatorname{DD}}_{h,k} + B(h) + \tau(h,k,n) \\
    & \leq \tilde{W}^{\operatorname{DD}}_{h}(\pi^*_k) - \hat{W}^{\operatorname{DD}}_{h}(\pi^*_k) + \tilde{R}^{\operatorname{DD}}_{h,k} + \hat{R}^{\operatorname{DD}}_{h,k} + 2B(h) + \tau(h,k,n),
\end{align*}
where $\tilde{W}^{\operatorname{DD}}_{h}$ is the infeasible welfare constructed using the true nuisance parameters
\begin{align*}
    \tilde{W}^{\operatorname{DD}}_{h}(\pi) \equiv \frac{1}{n}\sum_{i=1}^n \Gamma_h(Y_i,T_i,X_i;\pi;g,m).
\end{align*}
Given Lemma \ref{lm:DR-calculation}, the mean of $\tilde{W}^{\operatorname{DD}}_{h}(\pi^*_k) - \hat{W}^{\operatorname{DD}}_{h}(\pi^*_k)$ is of order $o(1)B(h) + o(n^{-(\rho_g + \rho_m)})$. The bound on $\mathbb{E}[\tilde{R}^{\operatorname{DD}}_{h,k}] = \mathbb{E}[\tilde{R}^{\operatorname{DD},\ell}_{h,k}]$ is derived in Lemma \ref{lm:tilde-R}, as $(C'_v+o(1))\sqrt{L\frac{\operatorname{VC}(\Pi_k)}{nh}}$. The deviations $\tilde{R}^{\operatorname{DD}}_{h,k}-\mathbb{E}[\tilde{R}^{\operatorname{DD}}_{h,k}]$ from the mean can be negligible in the same way as in the proof of Theorem \ref{thm:fully} by using the Talagrand inequality. The relevant class of functions is
\begin{align*}
    \left\{(Y,T,X,\text{Rad})\mapsto 2\text{Rad}\Gamma_h(Y,T,X;\pi;g,m),\pi\in\Pi_k\right\}.
\end{align*}
A uniform bound of this class of functions is $2\left(M+\frac{2\bar{\kappa} M}{h\underline{f}}\right) \asymp 1/h$, and a second-moment bound is $4\left(5M^2 + \frac{4M^2\kappa_2}{h\underline{f}}\right)\asymp 1/h$. The deviation from the mean can be bounded as
\begin{align*}
    \mathbb{P}\left(\tilde{R}^{\operatorname{DD}}_{h,k}-\mathbb{E}[\tilde{R}^{\operatorname{DD}}_{h,k}]>\alpha\right) &\leq C'\exp\left(-\frac{n\alpha^2}{\frac{C_1'}{h} + \frac{C_2'}{h}\sqrt{\frac{\operatorname{VC(\Pi_k)}}{nh^2}}+\frac{C_3'\alpha}{h}}\right) \\
    &\leq C'\exp\left(-\frac{nh_{min}\alpha^2}{C'+C'\alpha}\right)
\end{align*}
where $C_1'$, $C_2'$ and $C_3'$ involve constants in the above uniform and second-moment bounds. Note that for the bound on $\mathbb{E}[\tilde{R}^{\operatorname{DD},\ell}_{h,k}]$, we take the first bound (\ref{eqn:EDelta-bound1}) in Lemma \ref{lm:kitagawa}. Applying Lemma \ref{lm:tail-to-Op-rate}, we obtain that $\tilde{R}^{\operatorname{DD}}_{h,k}-\mathbb{E}[\tilde{R}^{\operatorname{DD}}_{h,k}]=O_p((nh_{min})^{-1/2})$.

The next task is to bound the feasible Rademacher complexity term (constructed using the estimated nuisance parameters) $\hat{R}^{\operatorname{DD},\ell}_{h,k}$:
    \begin{align*}
        \hat{R}^{\operatorname{DD},\ell}_{h,k}
        & = \underbrace{\sup_{\pi \in \Pi_k} \frac{L}{n} \sum_{i \in I_\ell} 2 \text{Rad}_i \Gamma_h(Y_i,T_i,X_i;\pi;g,m) }_{ \tilde{R}^{\operatorname{DD}}_{h,k}} \\
        & \quad + \sup_{\pi \in \Pi_k} \frac{L}{n} \sum_{i \in I_\ell} 2 \text{Rad}_i (\Gamma_h(Y_i,T_i,X_i;\pi;\hat{g}_\ell,\hat{m}_\ell) - \Gamma_h(Y_i,T_i,X_i;\pi;g,m)) .
    \end{align*}
    The term $\hat{R}^{\operatorname{DD},\ell}_{h,k}$ is already taken care of. To bound the remaining term, notice that, in view of Lemma \ref{lm:DR-calculation}, we can write
    \begin{align*}
        & \sup_{\pi \in \Pi_k} \frac{L}{n} \sum_{i \in I_\ell} 2 \text{Rad}_i (\Gamma_h(Y_i,T_i,X_i;\pi;\hat{g}_\ell,\hat{m}_\ell) - \Gamma_h(Y_i,T_i,X_i;\pi;g,m)) \\
        \leq & \sup_{\pi \in \Pi_k} \frac{L}{n} \sum_{i \in I_\ell} 2 \text{Rad}_i \Gamma_{1h}  +\sup_{\pi \in \Pi_k} \frac{L}{n} \sum_{i \in I_\ell} 2 \text{Rad}_i \Gamma_{2h}+ \sup_{\pi \in \Pi_k} \frac{L}{n} \sum_{i \in I_\ell} 2 \text{Rad}_i \Gamma_{3h} ,
    \end{align*}
    where $\Gamma_{1h}$, $\Gamma_{2h}$, and $\Gamma_{3h}$ are given in Lemma \ref{lm:DR-calculation}, and for simplicity, their expressions are suppressed. Denote the above three terms on the right-hand side as $\Delta_{1h\ell}$, $\Delta_{2h\ell}$, and $\Delta_{3h\ell}$ respectively.
    % Consider the following classes of functions for $j=1,2,3$, respectively,
    % \begin{align*}
    %     \{(Y,T,X,\text{Rad})\}\mapsto 2\text{Rad}\Gamma_{jh}(Y,T,X;\pi;g,m),\pi\in\Pi_k\}
    % \end{align*}
    
    By Lemma \ref{lm:Rhat-R}, the condition mean (given the estimators $\hat{g}_\ell$ and $\hat{m}_\ell$) of the above sum of three terms is bounded by
    \begin{align*}
        & \left(C_1 \frac{\lVert \hat{g}_\ell - g \rVert_\infty \lVert \hat{m}_\ell - m \rVert_\infty}{\sqrt{h}} + C_2 \frac{\lVert \hat{g}_\ell - g \rVert_\infty}{\sqrt{h}} + C_3 \frac{\lVert \hat{m}_\ell - m \rVert_\infty}{\sqrt{h}} \right) \sqrt{\frac{\operatorname{VC}(\Pi_k)}{(n/L)h}}\\
        = & o(1)\sqrt{L\frac{\operatorname{VC}(\Pi_k)}{nh}},
    \end{align*}
    where the sum in the parenthesis is $o(1)$ because we assume that $n^{-\rho_g},n^{-\rho_m} = o(n^{-r/(4r+2)})$ and $h \geq n^{-1/(2r+1)}$. 
    
    For the first term, its uniform bound is $\frac{2\kappa}{h}\|\hat{g}_\ell-g\|_\infty\|\hat{m}_\ell-m\|_\infty$, and the second moment bound is $\frac{4\kappa_2}{h\underline{f}}\|\hat{m}_\ell-m\|_\infty^2\|\hat{g}_\ell-h\|_\infty^2$. For the second term, the uniform bound is $\frac{4M\kappa}{h}\|\hat{g}_\ell-g\|_\infty$, and the second moment bound is $\frac{16M^2\kappa_2}{h\underline{f}}\|\hat{g}_\ell-g\|_\infty^2$. For the third term, the uniform bound is $2\left(\frac{\kappa}{h\underline{f}}-1\right)\|\hat{m}_\ell-m\|_\infty$, and the second moment bound is $4\left(\frac{\kappa_2}{h\underline{f}}-1\right)\|\hat{m}_\ell-m\|_\infty^2$. Then, the deviations from the mean of the three terms, $\Delta_{jh\ell}-\mathbb{E}[\Delta_{jh\ell}],j=1,2,3$, can be bounded, respectively, using the Talagrand inequality.
    \begin{align*}
        \mathbb{P}\left(\Delta_{jh\ell}-\mathbb{E}[\Delta_{jh\ell}]>\alpha \mid \hat{g}_\ell,\hat{m}_\ell\right) \leq C'_{j\ell}\exp\left(-\frac{n\alpha^2}{\frac{C'_{1jh\ell}}{h}+\frac{C'_{2jh\ell}}{h}\sqrt{\frac{\operatorname{VC}(\Pi_k)}{nh^2}}+\frac{C'_{3jh\ell}\alpha}{h}}\right),
    \end{align*}
for $j=1,2,3$, and the constant terms depend on the constants in the above uniform bounds and second-moment bounds. Same as the argument for the proof of Theorem \ref{thm:fully}, given the estimators $\hat{g}_\ell$ and $\hat{m}_\ell$, $\Delta_{jh\ell}-\mathbb{E}[\Delta_{jh\ell}] = O_p((nh_{min})^{-1/2})=O_p(n^{-r/(2r+1)})$ since $h_{min}= n^{-1/(2r+1)}$.

    Putting the above results together, we obtain the leading term for the double-debiased welfare regret:
    \begin{align*}
    & W^*(\Pi_\infty) - \hat{Q}^{\operatorname{DD}}_{h,k} \\
    \leq & W^*(\Pi_\infty) - W^*(\Pi_k) + 2 (C_v' + o(1)) \sqrt{L \frac{\operatorname{VC}(\Pi_k)}{nh}} + (2+o(1))B(h) + O_p(n^{-r/(2r+1)}).
    \end{align*}
    The remainder term $\hat{Q}_{h,k}^{\operatorname{DD}} - W(\hat{\pi}^{\operatorname{DD}})$ can be bounded following the steps in the proof of Theorem \ref{thm:fully}. The only difference is that we must take care of the difference $\hat{W}^{\operatorname{DD}}_{h,k} - \tilde{W}^{\operatorname{DD}}_{h,k}$, which converts to bounding the following tail.
    \begin{align*}
        & \sum_{h,k} \mathbb{P} \left( \hat{W}^{\operatorname{DD}}_{h,k} - \tilde{W}^{\operatorname{DD}}_{h,k} - \mathbb{E}[\hat{W}^{\operatorname{DD}}_{h,k} - \tilde{W}^{\operatorname{DD}}_{h,k}] >\alpha + \tau(h,k,n) \right) \\
        \leq & \sum_{j=1}^3 \sum_{h,k} \mathbb{P} \left( \hat{\Gamma}_{jh} - \mathbb{E}[\hat{\Gamma}_{jh}] > (\alpha + \tau(h,k,n))/3 \right),
    \end{align*}
    where $\hat{\Gamma}_{jh}$ is the sample mean of $\Gamma_{jh},j=1,2,3,$ as defined in Lemma \ref{lm:DR-calculation}. These tail probabilities are controlled in the same way using Talagrand's inequality, as
    \begin{align*}
        \mathbb{P}\left(\hat{\Gamma}_{jh}-\mathbb{E}[\hat{\Gamma}_{jh}]>(\alpha+\tau(h,k,n))/3\right) &\leq C''_j\exp\left(-\frac{n(\alpha+\tau(h,k,n))^2}{\frac{C''_{1jh}}{h}+\frac{C''_{2jh}}{h}\sqrt{\frac{\operatorname{VC}(\Pi_k)}{nh^2}}+\frac{C''_{3jh}(\alpha+\tau(h,k,n))}{h}}\right) \\
        &\leq C''_j\exp\left(-\frac{n(\alpha+\tau(h,k,n))^2}{C''_j + C''_j(\alpha+\tau(h,k,n))}\right).
    \end{align*}
    Taking the same argument as in the proof of Theorem \ref{thm:fully}, under condition (\ref{eqn:tau-requirement}), we obtain that the remainder term $\hat{Q}_{h,k}^{\operatorname{DD}}-W(\hat{\pi}^{\operatorname{DD}})$ is of order $O_p((nh_{min})^{-1/2})$. Combining the leading term and the remainder term completes the proof.
\end{proof}

\begin{lemma} \label{lm:monotone-sieve-approx}
    Assume that the marginal density of $X$ is bounded. For the monotone policy class introduced by Section \ref{sec:monotone-sieve}, the sieve approximation rate in condition (\ref{eqn:deficiency-rate}) can be taken to be $\alpha_k = O(k^{-1})$.
\end{lemma}

\begin{proof}[Proof of Lemma \ref{lm:monotone-sieve-approx}]
Without loss of generality, let $d_X=1$. The case with $d_X>1$ can be proved in the same way. 
For simplicity, assume $\pi^*=\sum_{p=1}^{d_X} h_p^*(x_p)$, the optimal policy in the global policy class $\Pi_\infty$, exists. That is, $\pi^*$ maximize $W(\pi)$ for $\pi \in \Pi_\infty$. Let $\pi^*_k \in \Pi_k$ denote the piecewise linear approximation of $\pi^*$. That is, on each endpoint $k'/k$, we have $\pi^*_k = \pi^*$, and $\pi_k^*$ is linear on each interval $[(k'-1)/k,k'/k], \forall k^\prime=1,\dots,k$. Denote $\bar{f}_X$ as the upper bound on the marginal density of $X$, $f_X$. The $L_1$ distance between $\pi^*$ and $\pi^*_k$ is bounded as
    \begin{align*}
    \E\left[|\pi^*(X) -\pi_k^*(X)|\right] & = \sum_{k^\prime=1}^k \int_{\frac{k^\prime-1}{k}}^\frac{k^\prime}{k} |\pi^*\left(x\right)-\pi^*\left(x\right)|f(x)dx \\
     &\leq \sum_{k^\prime=1}^k \left(\pi^*\left(\frac{k^\prime}{k}\right)-\pi^*\left(\frac{k^\prime-1}{k}\right)\right)\int_{\frac{k^\prime-1}{k}}^\frac{k^\prime}{k} f_X(x)dx \text{ (monotonicity of $\pi^*,\pi^*_k$)} \\
     &\leq \frac{\bar{f}_X}{k}\sum_{k^\prime=1}^k \left(\pi^*\left(\frac{k^\prime}{k}\right)-\pi^*\left(\frac{k^\prime-1}{k}\right)\right) \text{ ($f_X$ bounded)} \\
     & \leq \frac{\bar{f}_X}{k} \left(\pi^*\left(\sup \text{supp}(X)\right)-\pi^*\left(\inf \text{supp}(X)\right)\right) \text{ (monotonicity of $\pi^*,\pi^*_k$)},
     \end{align*}
     where in the last line, $\sup\operatorname{supp}(X)$ and $\inf\operatorname{supp}(X)$ denote the upper and lower endpoints of the support of $X$, respectively. This term on the right-hand side is $O(1/k)$ because $\pi^{*}\in\Pi_{\infty}$ is bounded, given that the treatment variable has compact support by Assumption \ref{ass:bounded}.

\end{proof}

\section{Additional Theoretical Discussion}

\subsection{Sufficient conditions for Assumption \ref{ass:rates-nonparametric-estimates}} \label{sec:nonparametric-estimators}
In this section, we describe concrete estimators for $g$ and $m$ and show how the high-level Assumption \ref{ass:rates-nonparametric-estimates} in Section \ref{sec:dr} can be satisfied.

\paragraph{Conditional density estimator} Consider the local polynomial conditional density estimator, as proposed by \cite{cattaneo2024boundary}:
\begin{align*}
    \hat{f}(t|x)=e_1'\hat{\beta}(t|x),\quad \hat{\beta}(t|x)=\argmin_{u\in \mathbb{R}^{p+1}} \sum_{i=1}^n\left(\hat{F}(T_i|x)-p(T_i-t)'u\right)^2K_h(T_i;t),
\end{align*}
where $p$ is the order of the polynomial basis $p(t) = (1,t/1!,t^2/2!,\cdots,t^p/p!)'$, $e_1$ is the unit vector, $K_h(T_i,t)=K((T_i-t)/h)/h$ for some kernel function $K$, and 
\begin{align*}
    \hat{F}(t|x)=e_0'\hat{\gamma}(t|x), \quad \hat{\gamma}(t|x)=\argmin_{v\in \mathbb{R}^{q_{d_X}+1}}\sum_{i=1}^n\left(\mathbf{1}(T_i\leq t)-q(X_i-x)'v\right)^2L_b(X_i;x),
\end{align*}
where $q(x)$ denotes the $(q_{d_X}+1)$-dimensional vector of terms $x^m/m!=x_1^{m_1}\cdots x_{d_X}^{m_{d_X}}$ for $x = (x_1,\cdots,x_m)$, $m=(m_1,\cdots,m_{d_X})$ with $|m|=m_1+\dots+m_{d_X}\leq q$, $q_{d_X}=({d_X}+q)!/({d_X}!q!)-1$, and $L_b(X_i,x)=L((X_i-x)/b)/b^{d_X}$ for some multivariate kernel function $L$.

According to Theorem 1 in \cite{cattaneo2024boundary}, if the joint density of $(T,X)$, $f(t,x)$, is continuous and bounded away from zero; if the conditional density $f(t|x)$ exists, is continuous, and has continuous $p$-th partial derivatives with respect to $x$; and if the kernel $K$ is symmetric, Lipschitz continuous, and support on $[0,1]$, then, as $h\to0$ and $nh^{1+d}/\log n\to\infty$,
\begin{align*}
    \sup_{t\in\mathcal{T},x\in\mathcal{X}}\left|\hat{f}(t|x)-f(t|x)\right| = h^p + \sqrt{\frac{\log n}{nh^{1+d_X}}}, 
\end{align*}
which is of order $O_p\left(\frac{\log n}{n}\right)^{\frac{p}{1+d_X+2p}}$ when choosing $h=\left(\frac{\log n}{n}\right)^{\frac{1}{1+d_X+2p}}$. By Assumption \ref{ass:bounded}, $f(t|x)$ is bounded below by some constant $c>0$. The uniform rate on $\hat{f}(t|x)$ can be transformed to the uniform rate of $\hat{g}(t,x) = 1/\hat{f}(t|x)$:
$$
    \sup_{t \in \mathcal{T},x \in \mathcal{X}} \left| \hat{g}(t,x) - g(t,x)\right| \leq  \frac{\sup_{t\in\mathcal{T},x\in\mathcal{X}}\left|\hat{f}(t|x)-f(t|x)\right|}{c (c - \lVert \hat{f} - f\rVert_\infty)} = O_p\left(\frac{\log n}{n}\right)^{\frac{p}{1+d_X+2p}},
$$
where we have used the fact that $(1+o_p(1))^{-1} = O_p(1)$ from Section 2.2 in \cite{vandervaart1998asymptotic}.

Therefore, when the smoothness of $f(t|x)$ satisfies that $p> \frac{r(1+d_X)}{2(1+r)}$, the rate requirement on $\hat{g}$ in Assumption \ref{ass:rates-nonparametric-estimates}(i) can be satisfied. Also, $\hat{g}$ is bounded with probability approaching one, as required by Assumption \ref{ass:rates-nonparametric-estimates}(ii), by the uniform consistency.

\paragraph{Conditional mean estimator} Consider the partitioning-based estimator for $m$ as proposed by \cite{CattaneoFengShigida2024uniform}. Assume that $m$ is $p$ times continuously differentiable in $(t,x)$. Here, $p$ may be smaller than $r$ in Assumption \ref{ass:smoothness-r} as $r$ only characterizes the smoothness of $m$ in $t$. 

The estimator is constructed as follows. First, partition the support of $(T,X)$ into hyper-rectangles of side length at most $h$. On each cell, form $p-\text{th}$-order piecewise-polynomial basis functions $p_k(t,x)\quad(k=1,\dots,K)$, with support confined to a few neighbouring cells.  The dictionary size grows like $K\;\asymp\;h^{-(1+d_X)}$. Then, we estimate the coefficients by least squares
   $$
   \hat{\beta}\in\arg\min_{b\in\mathbb{R}^K}\;
      \sum_{i=1}^n\bigl(Y_i-p(Z_i)^{\!\top}\!b\bigr)^2,
      \qquad Z_i=(T_i,X_i).
   $$
The plug-in predictor is $\hat{m}(t,x)\;=\;p(t,x)^{\!\top}\hat{\beta}$.

Under Assumptions 1–6 in \cite{CattaneoFengShigida2024uniform}, choosing $h\asymp\bigl(\frac{\log n}{n}\bigr)^{\frac{1}{2p+1+d_X}}$ implies, by their Corollary 1, that
$$
\sup_{t,x}\lvert\hat m(t,x)-m(t,x)\rvert=O_p\!\bigl(\tfrac{\log n}{n}\bigr)^{\frac{p}{2p+1+d_X}}.
$$
Therefore, when the smoothness parameter of $m(t,x)$ satisfies $p>\frac{r(1+d_X)}{2(1+r)}$, the rate requirement in Assumption \ref{ass:rates-nonparametric-estimates}(i) holds. Moreover, uniform consistency ensures that $\hat m$ is bounded with probability approaching one, satisfying Assumption \ref{ass:rates-nonparametric-estimates}(ii). Finally, such partitioning-based estimators are of bounded variation for each fixed partition, fulfilling Assumption \ref{ass:rates-nonparametric-estimates}(iii).

\subsection{Searching across $k \in \mathbb{N}$} \label{sec:k-infty}

The data-driven procedure restricts attention to values of $k$ satisfying $\operatorname{VC}(\Pi_k)\le nh^{2}$. Although an exhaustive search over all $k\ge1$ (until $\infty$) is impractical, it highlights an interesting theoretical point worth discussing.

The point at which the proof of Theorem \ref{thm:fully} breaks down is the tail bound in the last line of (\ref{eqn:tail-bound}). Under the restriction $\operatorname{VC}(\Pi_k) \le nh^2$, the denominator in the exponential bound simplifies to $C + C(\alpha + \tau)$, with the term $\sqrt{\frac{\operatorname{VC}(\Pi_k)}{nh^2}}$ absorbed into the constant. However, when $k$ is allowed to grow without bound, the term $\sqrt{\frac{\operatorname{VC}(\Pi_k)}{nh^2}}$ may become dominant in the denominator. In that case, the exponential tail takes the following form:
\begin{align*}
    & \exp \left( - \frac{nh (\alpha + \tau(h,k,n))^2}{C \sqrt{\frac{\operatorname{VC}(\Pi_k)}{nh^2}}} \right) \\
    = & \exp \left( - \frac{nh \alpha^2}{C \sqrt{\frac{\operatorname{VC}(\Pi_k)}{nh^2}}} \right) \underbrace{\exp \left( - \frac{2nh \alpha\tau(h,k,n)}{C \sqrt{\frac{\operatorname{VC}(\Pi_k)}{nh^2}}} \right)}_{\leq 1} \exp \left( - \frac{nh \tau(h,k,n)^2}{C \sqrt{\frac{\operatorname{VC}(\Pi_k)}{nh^2}}} \right).
\end{align*}
Now we integrate the tail probability over $\alpha$ to pin down the rate for the first term on the right-hand side:
\begin{align*}
    \int_0^\infty \mathbb{P}(\Delta_{h,k} - \bar{\Delta}_{h,k} - \tau(h,k,n)/2 > \alpha) d \alpha \leq \left[\int_0^\infty \exp \left( - \frac{nh \alpha^2}{C \sqrt{\frac{\operatorname{VC}(\Pi_k)}{nh^2}}} \right) d\alpha \right] \exp \left( - \frac{nh \tau(h,k,n)^2}{C \sqrt{\frac{\operatorname{VC}(\Pi_k)}{nh^2}}} \right).
\end{align*}
By a standard integration argument \citep[see, e.g., Problem 12.1 in][]{devroye1996probabilistic}, the integral on the right-hand side is bounded as
\begin{align*}
    \int_0^\infty \exp \left( - \frac{nh \alpha^2}{C \sqrt{\frac{\operatorname{VC}(\Pi_k)}{nh^2}}} \right) d\alpha \leq C \frac{\operatorname{VC}(\Pi_k)^{1/4}}{n^{3/4}h} \leq C \frac{\operatorname{VC}(\Pi_k)^{1/4}}{n^{3/4}h_{min}}.
\end{align*}
Therefore, we have
\begin{align} \label{eqn:k-infty-calculation1}
    \mathbb{E}[\Delta_{\hat{h},\hat{k}} - \bar{\Delta}_{\hat{h},\hat{k}} - \tau(\hat{h},\hat{k},n)/2] & \leq \int_0^\infty \mathbb{P}(\Delta_{\hat{h},\hat{k}} - \bar{\Delta}_{\hat{h},\hat{k}} - \tau(\hat{h},\hat{k},n)/2>\alpha) d\alpha \nonumber \\
    & \leq \sum_{h \in \mathcal{H}, k \geq 1} \int_0^\infty \mathbb{P}(\Delta_{h,k} - \bar{\Delta}_{h,k} - \tau(h,k,n)/2 > \alpha) d \alpha \nonumber \\
    & \leq C \sum_{h \in \mathcal{H}, k \geq 1} \frac{\operatorname{VC}(\Pi_k)^{1/4}}{n^{3/4}h_{min}} \exp \left( - \frac{n^{3/2}h^2 \tau(h,k,n)^2}{C \sqrt{\operatorname{VC}(\Pi_k)}} \right).
\end{align}
where the first line is a standard representation of the expectation by integrating the tail probability, and the second line uses the union bound. For simplicity, we treat $h_{min}$ as non-random as it can be estimated from an independent sample. Observe that $h_{min} \geq n^{-1/(2\hat{r}+1)} \geq n^{-1/3}$ since the order of smoothness $r$ is at least 1 by Assumption \ref{ass:smoothness-r}. Then we have $(n^{3/4}h_{min})^{-1} = o((nh_{min})^{-1/2})$ because
\begin{align*}
    \frac{n^{3/4}h_{min}}{\sqrt{nh_{min}}} = \sqrt{n^{1/2}h_{min}} \rightarrow \infty.
\end{align*}
Therefore, the right-hand side of (\ref{eqn:k-infty-calculation1}) can be made smaller than $(nh_{min})^{-1/2}$ provided that the following double sum is finite and stays bounded as $n$ grows.
\begin{align*}
    \sum_{h \in \mathcal{H}, k \geq 1} \operatorname{VC}(\Pi_k)^{1/4}\exp \left( - \frac{n^{3/2}h^2 \tau(h,k,n)^2}{C \sqrt{\operatorname{VC}(\Pi_k)}} \right).
\end{align*}
Using once again the fact that $n^{1/2}h_{min} \rightarrow \infty$, the above double sum is bounded by
\begin{align*}
    \sum_{h \in \mathcal{H}, k \geq 1} \operatorname{VC}(\Pi_k)^{1/4}\exp \left( - \frac{nh \tau(h,k,n)^2}{C \sqrt{\operatorname{VC}(\Pi_k)}} \right).
\end{align*}
This sum can be finite and bounded as $n$ grows if we choose 
\begin{align*}
    \tau(k,h,n) = \sqrt{\frac{\operatorname{VC}(\Pi_k)^{1/2}}{nh} (\lambda_k\log (\operatorname{VC}(\Pi_k)) + \lambda_k \log k - \lambda'_h \log h)},
\end{align*}
for any sequences $\lambda_k,\lambda'_h \rightarrow \infty$. Consider the bandwidth set $\mathcal{H} \subset \{h = j^{-\rho}, j \in \mathbb{N}_+\}$. Then the double sum is bounded above by
\begin{align*}
    & \sum_{h \in \mathcal{H}, k \geq 1} \operatorname{VC}(\Pi_k)^{1/4}\exp \left( - \frac{nh \frac{\operatorname{VC}(\Pi_k)^{1/2}}{nh} (\lambda_k\log (\operatorname{VC}(\Pi_k)) + \lambda_k \log k - \lambda'_h \log h)}{C \sqrt{\operatorname{VC}(\Pi_k)}} \right) \\
    = & \sum_{h \in \mathcal{H}, k \geq 1} \underbrace{\operatorname{VC}(\Pi_k)^{1/4}\exp \left(-\lambda_k\log (\operatorname{VC}(\Pi_k)/C) \right)}_{\leq 1} \underbrace{\exp(-(\lambda_k \log k)/C)}_{\leq k^{-2}} \underbrace{\exp((\lambda'_h \log h)/C)}_{\leq h^{2/\rho}} \\
    \leq & C \left( \sum_{k \geq 1} 1/k^2 \right) \left( \sum_{j \geq 1} j^{-2}\right),
\end{align*}
where the inequalities in the underbraces hold when $k$ is sufficiently large and when $h$ is sufficiently small. 

The derivation above shows that, when optimizing over $k \in \mathbb{N}_{+}$, we need to increase $\tau$ from logarithmic order to a polynomial order in $\operatorname{VC}(\Pi_k)$. Nevertheless, $\tau$ can remain of smaller order than the variance term $\frac{\operatorname{VC}(\Pi_k)}{n h}$.

\section{Auxiliary results} \label{sec:aux}
\subsection{Additional lemmas}
This section provides auxiliary lemmas for the proofs in the previous section. In particular, we derive results for the VC dimension, the uniform bound, and the second-moment bound for the relevant function classes.

\begin{lemma} \label{lm:kernel-VC-dim}
    Let Assumption \ref{ass:infinite-kernel} hold. For any $h > 0$ and $k \geq 1$, the VC dimension of the function class
    \begin{align*}
        \left\{ (Y,T,X) \mapsto \frac{1}{h} K \left( \frac{T-\pi(X)}{h} \right) \frac{Y}{f(T|X)} : \pi \in \Pi_k\right\}
    \end{align*}
    is bounded by $2 \operatorname{VC}(\Pi_k)$. The same VC dimension bound also applies to the function class associated with the Rademacher complexity:
    \begin{align*}
        \left\{ (\text{Rad},Y,T,X) \mapsto 2\text{Rad}\frac{1}{h} K \left( \frac{T-\pi(X)}{h} \right) \frac{Y}{f(T|X)} : \pi \in \Pi_k\right\}.
    \end{align*}
\end{lemma}

\begin{proof}[Proof of Lemma \ref{lm:kernel-VC-dim}]
    The mapping $(Y,T,X) \mapsto \frac{Y}{hf(T|X)}$ is a fixed function. Given Lemma \ref{lm:VC-permanence}, we only need to bound the VC dimension of the following class:
    \begin{align*}
        \left\{ (Y,T,X) \mapsto K \left( \frac{T-\pi(X)}{h} \right) : \pi \in \Pi_k\right\}.
    \end{align*}
    By Assumption \ref{ass:infinite-kernel} and Lemma 3.6.11 in \cite{gine2021mathematical}, $K = K_1 \circ K_2$, where $K_1$ is a Lipschitz continuous function and $K_2$ is a nondecreasing function. Given Lemma \ref{lm:VC-permanence}, we only need to bound the VC dimension of the following class:
    \begin{align*}
        \left\{ (Y,T,X) \mapsto K_2 \left( \frac{T-\pi(X)}{h} \right) : \pi \in \Pi_k\right\}.
    \end{align*}
    Denoting $K_2^{-1}$ as the generalized inverse of the non-decreasing function $K_2$, we can write the subgraph of a function in the above class as
    \begin{align*}
        \left\{ (\Upsilon,Y,T,X): \Upsilon \leq K_2 \left( \frac{T-\pi(X)}{h}  \right) \right\} = \left\{ (\Upsilon,Y,T,X): \pi(X) + hK_2^{-1}(\Upsilon) - T \leq 0 \right\},
    \end{align*}
    which is the negative set of the function $(\Upsilon,Y,T,X) \mapsto \pi(X) + hK_2^{-1}(\Upsilon) - T$. Since $hK_2^{-1}(\Upsilon) - T$ is a fixed function, the function class $\{(\Upsilon,Y,T,X) \mapsto \pi(X) + hK_2^{-1}(\Upsilon) - T:\pi \in \Pi_k\}$ has the same VC dimension as $\Pi_k$ (Lemma \ref{lm:VC-permanence}). Its negative set also has the same VC dimension given the proof of Lemma 2.6.18(iii) in \cite{wellner1996}. In the end, the above derivation gives the desired result. The second claim of the lemma also follows from Lemma \ref{lm:VC-permanence} by treating $(\text{Rad},Y,T,X) \mapsto \frac{2\text{Rad}Y}{hf(T|X)}$ as a fixed function.
\end{proof}

\begin{lemma} \label{lm:uniform-variance-bound}
    Let Assumptions \ref{ass:bounded} and \ref{ass:infinite-kernel} hold. For any $h >0$ and $k \geq 1$, the following function class
    \begin{align*}
        \mathcal{G}_k \equiv \left\{ (Y,T,X) \mapsto \frac{1}{h} K \left( \frac{T-\pi(X)}{h} \right) \frac{Y}{f(T|X)} : \pi \in \Pi_k\right\}
    \end{align*}
    admits the uniform bound $B(\mathcal{G}_k) = \frac{\bar{\kappa} M}{h \underline{f}}$ and the second-moment bound $\sigma^2(\mathcal{G}_k) = \frac{M^2 \kappa_2}{h \underline{f}}$. 
\end{lemma}

\begin{proof}[Proof of Lemma \ref{lm:uniform-variance-bound}]
    The uniform bound is evident from the boundedness of $Y$, $f$, and $K$. For the second-moment bound, notice that
    \begin{align*}
        \mathbb{E}\left[ \left( \frac{1}{h} K \left( \frac{T - \pi(X)}{h} \right) \frac{Y}{f(T|X)} \right)^2 \right] = & \frac{1}{h^2} \mathbb{E}\left[  K \left( \frac{T - \pi(X)}{h} \right)^2\frac{Y^2}{f(T|X)^2}  \right] \\
        \leq & \frac{M^2}{h^2} \mathbb{E}\left[  K \left( \frac{T - \pi(X)}{h} \right)^2 \frac{1}{f(T|X)^2} \right].
    \end{align*}
    By change of variables, we have 
    \begin{align*}
         \mathbb{E}\left[  K \left( \frac{T - \pi(X)}{h} \right)^2 \frac{1}{f(T|X)^2}  \right] = & \int \int K\left( \frac{t - \pi(x)}{h} \right)^2 \frac{1}{f(t|x)} dt f_X(x) dx \\
         = & h \int \int K\left( v \right)^2 \frac{1}{f(\pi(x)+hv|x)} dv f_X(x) dx \\
         \leq & \frac{h}{\underline{f}} \int K(v)^2 dv \int f_X(x)dx = \frac{\kappa_2h}{\underline{f}} .
    \end{align*}
    Therefore, the second moment is bounded as
    \begin{align*}
        \mathbb{E}\left[ \left( \frac{1}{h} K \left( \frac{T - \pi(X)}{h} \right) \frac{Y}{f(T|X)} \right)^2 \right] \leq \frac{M^2 \kappa_2}{h\underline{f}}.
    \end{align*}
\end{proof}

\begin{lemma} \label{lm:DR-IPW-same-bias}
    Under Assumption \ref{ass:infinite-kernel}, we have 
    \begin{align*}
        \mathbb{E}\left[ \left( 1 - \frac{1}{h f(T|X)} K \left( \frac{T-\pi(X)}{h} \right) \right)m(\pi(X),X) \right] = 0.
    \end{align*}
\end{lemma}

\begin{proof}[Proof of Lemma \ref{lm:DR-IPW-same-bias}]
    By the law of iterated expectations, it suffices to show that the following conditional expectation equals one almost surely:
    \begin{align*}
        \mathbb{E}\left[\frac{1}{h f(T|X)} K \left( \frac{T-\pi(X)}{h} \right)\Bigg| X\right].
    \end{align*}
    This holds by applying the standard change of variables.
\end{proof}

\begin{lemma} \label{lm:tilde-R}
    Given Assumptions \ref{ass:infinite-kernel} and \ref{ass:m-bounded-variation}, the expectation of the infeasible Rademacher complexity is bounded by
    \begin{align*}
        \mathbb{E}[\tilde{R}^{\operatorname{DD},\ell}_{h,k}] \leq ( C_{v}' + o(1)) \sqrt{L\frac{\operatorname{VC}(\Pi_k)}{nh}},
    \end{align*}
    where $C_v' \equiv (c+c')M \sqrt{\frac{\kappa_2}{\underline{f}}}$ and $c'$ is another universal constant derived from the proof.
\end{lemma}

\begin{proof}[Proof of Lemma \ref{lm:tilde-R}]
    By the definition of the double-debiased moment  function $\Gamma_h$ in (\ref{eqn:double-robustness}), the Rademacher complexsity $\tilde{R}^{\operatorname{DD},\ell}_{h,k}$ can be decomposed into two parts:
    \begin{align*}
        \mathbb{E}[\tilde{R}^{\operatorname{DD},\ell}_{h,k}] \leq & \mathbb{E}[\hat{R}^{\operatorname{IPW},\ell}_{h,k}] \\
        & + 2 \mathbb{E}\left[ \sup_{\pi \in \Pi_k} \frac{L}{n} \sum_{i \in I_\ell} \text{Rad}_i \left(1- \frac{1}{hf(T_i|X_i)} K \left( \frac{T_i-\pi(X_i)}{h} \right)  \right)m(\pi(X_i),X_i) \right],
    \end{align*}
    where $\hat{R}^{\operatorname{IPW},\ell}_{h,k}$ denotes the IPW Rademacher complexity used in Section \ref{sec:ipw} but constructed only using the cross-fitting method.
    The first term was already examined in the proof of Theorem \ref{thm:fully}. The second term is (twice of) the Rademacher complexity for the following class
    \begin{align*}
        \mathcal{F}_k \equiv \left\{ (Y,T,X) \mapsto \left(1- \frac{1}{hf(T|X)} K \left( \frac{T-\pi(X)}{h} \right)  \right)m(\pi(X),X): \pi \in \Pi_k \right\}.
    \end{align*}
    For this class, the uniform bound can be taken as $M(1 + \frac{\bar{\kappa}}{h \underline{f}})$, and the second-moment bound can be taken as $M^2 (1+ \frac{\kappa_2}{h\underline{f}})$. To bound the complexity of this class, note that both $K(\cdot)$ and $m(\cdot,x)$ are of bounded variation. Therefore, we can write
    \begin{align*}
        K(\cdot) = K_+(\cdot) - K_-(\cdot), m(\cdot,x) = m_+(\cdot,x) - m_-(\cdot,x),
    \end{align*}
    where $K_+(\cdot)$, $K_-(\cdot)$, $m_+(\cdot,x)$, and $m_-(\cdot,x)$ are nondecreasing functions. The function class $\mathcal{F}_k$ can be decomposed as
    \begin{align*}
        \mathcal{F}_k \subset (\mathcal{F}_{k,K+} \oplus \mathcal{F}_{k,K-})(\mathcal{F}_{k,m+} \oplus \mathcal{F}_{k,m-}),
    \end{align*}
    where
    \begin{align*}
        \mathcal{F}_{k,K+} & \equiv \left\{ (T,X) \mapsto 1- K_+\left( \frac{T-\pi(X)}{h} \right): \pi \in \Pi_k \right\}, \\
        \mathcal{F}_{k,K-} & \equiv \left\{ (T,X) \mapsto K_-\left( \frac{T-\pi(X)}{h} \right): \pi \in \Pi_k \right\}, \\
        \mathcal{F}_{k,m+} & \equiv \left\{ (T,X) \mapsto m_+(\pi(X),X) : \pi \in \Pi_k \right\}, \\
        \mathcal{F}_{k,m-} & \equiv \left\{ (T,X) \mapsto m_-(\pi(X),X) : \pi \in \Pi_k \right\}.
    \end{align*}
    The subgraph of a function in $\mathcal{F}_{k,K+}$ is 
    \begin{align*}
        \left\{ (\Upsilon,T,X) : \Upsilon \leq 1- K_+\left( \frac{T-\pi(X)}{h} \right) \right\} = \{ (\Upsilon,T,X) : \pi(X) + hK_+^{-1}(1-\Upsilon) - T \geq 0 \},
    \end{align*}
    which is the negative set of the function $(\Upsilon,T,X) \mapsto \pi(X) + hK_+^{-1}(1-\Upsilon) - T$ that $K_+^{-1}$ denotes the generalized inverse of the nondecreasing function $K_+$. Because $hK_+^{-1}(1-\Upsilon) - T$ is a fixed function, the VC dimension of this function class is the same as $\Pi_k$ (Lemma \ref{lm:VC-permanence}). Hence, the VC dimension of its negative set is also $\operatorname{VC}(\Pi_k)$. This means that $\operatorname{VC}(\mathcal{F}_{k,K+}) = \operatorname{VC}(\Pi_k)$. Similarly, we can show that the VC dimensions of $\mathcal{F}_{k,K-}$ $\mathcal{F}_{k,m-}$ are equal to those of $\Pi_k$. Then, by Lemma \ref{lm:covering-number-VC}, their covering numbers can be bounded as
    \begin{align*}
        \sup_{Q} N(\epsilon (1+\|K_+\|_\infty), \mathcal{F}_{k,K+}, L_2(Q)) \leq \left( \frac{c_0}{\epsilon} \right)^{2 \operatorname{VC}(\Pi_k)}, 
        \sup_{Q} N(\epsilon \|K_-\|_\infty, \mathcal{F}_{k,K-}, L_2(Q)) \leq \left( \frac{c_0}{\epsilon} \right)^{2 \operatorname{VC}(\Pi_k)}, \\
        \sup_{Q} N(\epsilon \|m_+\|_\infty, \mathcal{F}_{k,m+}, L_2(Q)) \leq \left( \frac{c_0}{\epsilon} \right)^{2 \operatorname{VC}(\Pi_k)}, 
        \sup_{Q} N(\epsilon \|m_-\|_\infty, \mathcal{F}_{k,m-}, L_2(Q)) \leq \left( \frac{c_0}{\epsilon} \right)^{2 \operatorname{VC}(\Pi_k)}.
    \end{align*}
    By Lemma \ref{lm:covering-number-permanence}, we can sum up the VC dimensions when operating element-wise addition and multiplication of function classes,
    \begin{align*}
        \sup_{Q} N(\epsilon (1+\bar{\kappa})M, \mathcal{F}_{k}, L_2(Q)) \leq & \sup_{Q} N(\epsilon (1+\|K_+\|_\infty/4), \mathcal{F}_{k,K+}, L_2(Q)) \\
        & \times \sup_{Q} N(\epsilon \|K_-\|_\infty/4, \mathcal{F}_{k,K-}, L_2(Q)) \\
        & \times \sup_{Q} N(\epsilon \|m_+\|_\infty/4, \mathcal{F}_{k,m+}, L_2(Q)) \\
        & \times \sup_{Q} N(\epsilon \|m_-\|_\infty/4, \mathcal{F}_{k,m-}, L_2(Q)) \leq \left( \frac{4c_0}{\epsilon} \right)^{8 \operatorname{VC}(\Pi_k)}.
    \end{align*}
    Then, by Lemma \ref{lm:kitagawa}, we have
    \begin{align*}
        & 2\mathbb{E}\left[ \sup_{\pi \in \Pi_k} \frac{L}{n} \sum_{i \in I_\ell} \text{Rad}_i \left(1- \frac{1}{hf(T_i|X_i)} K \left( \frac{T_i-\pi(X_i)}{h} \right)  \right)m(\pi(X_i),X_i) \right] \\
        \leq & 2\left( 2 c^2 M\left( 1 + \frac{\bar{\kappa}}{\underline{f}h} \right) \frac{4 \operatorname{VC}(\Pi_k)}{n/L} + c M\sqrt{1 + \frac{\kappa_2}{\underline{f}h}} \sqrt{\frac{4 \operatorname{VC}(\Pi_k)}{n/L}} \right) \\
        \leq & 2\left(2 c^2 M\left( h + \frac{\bar{\kappa}}{\underline{f}} \right) L\frac{4 \operatorname{VC}(\Pi_k)}{nh} + c M\left(\sqrt{h} + \sqrt{\frac{\kappa_2}{\underline{f}}}\right) \sqrt{L \frac{4 \operatorname{VC}(\Pi_k)}{nh}} \right)\text{ ($\because \sqrt{a+b} \leq \sqrt{a} + \sqrt{b}, \forall a,b>0$)} \\
        = & (C_{v}' + o(1)) \sqrt{L \frac{\operatorname{VC}(\Pi_k)}{nh}},
    \end{align*}
    where in the last line we have defined $c' = 4c$.
    This proves the desired results.
\end{proof}

\begin{lemma} \label{lm:DR-calculation}
Based on elementary calculations, the difference $\Gamma_h(Y,T,X;\pi;f,m) - \Gamma_h(Y,T,X;\pi;\hat{f},\hat{m})$ admits the following decomposition:
    \begin{align*}
        \Gamma_h(Y,T,X;\pi;f,m) - \Gamma_h(Y,T,X;\pi;\hat{f},\hat{m}) = \Gamma_{1h} + \Gamma_{2h} + \Gamma_{3h},
    \end{align*}
    where
    \begin{align*}
        \Gamma_{1h}(T,X;\pi;g,m;\hat{g}_\ell,\hat{m}_\ell) & \equiv \frac{1}{h}K\left( \frac{T-\pi(X)}{h} \right) (m(\pi(X),X) - \hat{m}_\ell(\pi(X),X))(\hat{g}_\ell(T,X) - g(T,X)), \\
        \Gamma_{2h}(Y,T,X;\pi;g,m;\hat{g}_\ell) & \equiv \frac{1}{h}K\left( \frac{T-\pi(X)}{h} \right) (Y - m(\pi(X),X))(\hat{g}_\ell(T,X) - g(T,X)) , \\
        \Gamma_{3h}(T,X;\pi;g,m;\hat{m}_\ell) & \equiv \left( \frac{1}{h}K\left( \frac{T-\pi(X)}{h} \right) g(T,X) - 1 \right)(m(\pi(X),X) - \hat{m}_\ell(\pi(X),X)).
    \end{align*}
    Denote $\mathbb{E}_\ell[\cdot] \equiv \mathbb{E}[\cdot|S^{(-\ell)}]$ as the conditional expectation operator given the data not in the $\ell$th fold. Under Assumptions \ref{ass:infinite-kernel}, \ref{ass:smoothness-r}, and \ref{ass:rates-nonparametric-estimates}, the conditional means of the above three terms satisfy that for any policy $\pi$,
    \begin{align*}
        0 \leq \mathbb{E}_\ell[\Gamma_{1h}(T,X;\pi;g,m;\hat{g}_\ell,\hat{m}_\ell)] \leq & \bar{f} \lVert \hat{m}_\ell - m \rVert_\infty \lVert \hat{g}_\ell - g \rVert_\infty, \\
        |\mathbb{E}_\ell[\Gamma_{2h}(Y,T,X;\pi;g,m;\hat{g}_\ell)]| \leq & C h^r \lVert \hat{g}_\ell - g \rVert_\infty, \\
        \mathbb{E}_\ell[\Gamma_{3h}(T,X;\pi;g,m;\hat{m}_\ell)] = & 0,
    \end{align*}
    where the expectation is taken with respect to the joint distribution of $(Y,T,X,
    \hat{g}_\ell,\hat{m}_\ell)$ with $(Y,T,X) \perp (\hat{g}_\ell,\hat{m}_\ell)$.
    %Under Assumptions \ref{ass:infinite-kernel} and \ref{ass:smoothness-r}, the second term in the above can be further bounded as
    %\begin{align*}
        %|\mathbb{E}[\Gamma_{2h}(Y,T,X;\pi;g,m;\hat{g}_\ell)]| \leq & O(h^r \mathbb{E}[\lVert \hat{g}_\ell - g \rVert_\infty]).
    %\end{align*}
\end{lemma}

\begin{proof}[Proof of Lemma \ref{lm:DR-calculation}]
    Notice that, due to the cross-fitting method, the sample in the $\ell$th fold remains iid given $S^{(-\ell)}$.
    For $\Gamma_{1h}$, we have 
    \begin{align*}
        \mathbb{E}_\ell[\Gamma_{1h}(T,X;\pi;g,m;\hat{g}_\ell,\hat{m}_\ell)] \leq \mathbb{E}_\ell\left[ \frac{1}{h}K\left( \frac{T-\pi(X)}{h} \right)\right] \lVert \hat{m}_\ell - m \rVert_\infty \lVert \hat{g}_\ell - g \rVert_\infty.
    \end{align*}
    The expectation on the RHS can be bounded as
    \begin{align*}
        \mathbb{E}_\ell \left[ \frac{1}{h}K\left( \frac{T-\pi(X)}{h} \right) \right] = & \int \frac{1}{h}K\left( \frac{t-\pi(x)}{h} \right)f(t|x) dt f_X(x)dx \\
        = & \int K(v) f(\pi(X)+vh|x) dv f_X(x)dx \leq \bar{f} \int K(v) dv f_X(x)dx = \bar{f}.
    \end{align*}
    Then, by the Cauchy-Schwartz inequality, we have 
    \begin{align*}
        \mathbb{E}_\ell[\Gamma_{1h}(T,X;\pi;g,m;\hat{g}_\ell,\hat{m}_\ell)] 
        \leq & \bar{f} \lVert \hat{m}_\ell - m \rVert_\infty \lVert \hat{g}_\ell - g \rVert_\infty.
    \end{align*} 
    For $\Gamma_{2h}$, by applying the law of iterated expectations and the standard change of variables, we have 
    \begin{align*}
        & |\mathbb{E}_\ell[\Gamma_{2h}(Y,T,X;\pi;g,m;\hat{g}_\ell)]|\\
        = & \left|\mathbb{E}_\ell \left[ \frac{1}{h}K\left( \frac{T-\pi(X)}{h} \right) (m(T,X) - m(\pi(X),X))(\hat{g}_\ell(T,X) - g(T,X)) \right]  \right| \\
        \leq & \mathbb{E}_\ell \left[ \frac{1}{h}K\left( \frac{T-\pi(X)}{h} \right) |m(T,X) - m(\pi(X),X)| \right] \lVert \hat{g}_\ell - g \rVert_\infty \\
        = & \int \frac{1}{h}K\left( \frac{t-\pi(x)}{h} \right) |m(t,x) - m(\pi(x),x)| f(t|x) dtf_X(x) dx \lVert \hat{g}_\ell - g \rVert_\infty \\
        \leq & C h^r \lVert \hat{g}_\ell - g \rVert_\infty,
    \end{align*}
    where the last line follows from Assumption \ref{ass:smoothness-r} and the fact that $K$ is an infinite-order kernel.
    %Therefore, the unconditional expectation is bounded as
    %\begin{align*}
        %|\mathbb{E}[\Gamma_{2h}(Y,T,X;\pi;g,m;\hat{g}_\ell)]| \leq \mathbb{E}[|\mathbb{E}_\ell[\Gamma_{2h}(Y,T,X;\pi;g,m;\hat{g}_\ell)]|] = C h^r \mathbb{E}[\lVert \hat{g}_\ell - g \rVert_\infty].
    %\end{align*}
    For $\Gamma_{3h}$, it has zero conditional mean because of Lemma \ref{lm:DR-IPW-same-bias}.
\end{proof}

\begin{lemma}\label{lm:Rhat-R}
    Under Assumptions \ref{ass:infinite-kernel}, \ref{ass:m-bounded-variation}, and \ref{ass:rates-nonparametric-estimates}, we have 
    \begin{align*}
        \mathbb{E}_\ell \left[ \sup_{\pi \in \Pi_k} \frac{L}{n} \sum_{i \in I_\ell} 2 \operatorname{Rad}_i \Gamma_{1h}(T_i,X_i;\pi;g,m;\hat{g}_\ell,\hat{m}_\ell) \right] \leq & C_1 \frac{\mathbb{E}\lVert \hat{g}_\ell - g \rVert_\infty \lVert \hat{m}_\ell - m \rVert_\infty}{h} \sqrt{L\frac{\operatorname{VC}(\Pi_k)}{n}}, \\
        \mathbb{E}_\ell \left[ \sup_{\pi \in \Pi_k} \frac{L}{n} \sum_{i \in I_\ell} 2 \operatorname{Rad}_i \Gamma_{2h}(Y_i,T_i,X_i;\pi;g,m;\hat{g}_\ell) \right] \leq & C_2 \frac{\lVert \hat{g}_\ell - g \rVert_\infty}{h} \sqrt{L\frac{\operatorname{VC}(\Pi_k)}{n}}, \\
        \mathbb{E}_\ell \left[ \sup_{\pi \in \Pi_k} \frac{L}{n} \sum_{i \in I_\ell} 2 \operatorname{Rad}_i \Gamma_{3h}(T_i,X_i;\pi;g,m;\hat{m}_\ell) \right] \leq & C_3 \frac{\mathbb{E}\lVert \hat{m}_\ell - m \rVert_\infty}{h} \sqrt{L\frac{\operatorname{VC}(\Pi_k)}{n}}, 
    \end{align*}
    where $C_1,C_2,C_3$ are constants that only depend on the kernel $K$ and the joint distribution of $(Y,T,X)$ and do not depend on $k$ or $n$. 
\end{lemma}

\begin{proof}[Proof of Lemma \ref{lm:Rhat-R}]
    We only prove the first inequality because the others can be proved analogously. We first bound the conditional expectation:
    \begin{align*}
        \mathbb{E} \left[ \sup_{\pi \in \Pi_k} \frac{L}{n} \sum_{i \in I_\ell} 2 \text{Rad}_i \Gamma_{1h}(T_i,X_i;\pi;g,m;\hat{g}_\ell,\hat{m}_\ell) \Bigg| \hat{g}_\ell,\hat{m}_\ell \right].
    \end{align*}
    Because, given $(\hat{g}_\ell,\hat{m}_\ell)$, the data points in $I_\ell$ remain an iid sample, we can use the empirical process theory in Appendix \ref{sec:empirical-process} to bound this conditional expectation. 
    Like the proof of Lemma \ref{lm:tilde-R}, we can show that the relevant function class satisfies the entropy bound (\ref{eqn:VC-type}) with some index $V$ proportional to $\operatorname{VC}(\Pi_k)$. The uniform bound for the relevant function class is proportional to $\frac{\lVert \hat{g}_\ell - g \rVert_\infty \lVert \hat{m}_\ell - m \rVert_\infty}{h}$ (notice that $\hat{g}_\ell$ and $\hat{m}_\ell$ are fixed for now). Then by (\ref{eqn:EDelta-bound1}) in Lemma \ref{lm:kitagawa}, the Rademacher complexity conditional on $(\hat{g}_\ell,\hat{m}_\ell)$ is bounded as
    \begin{align*}
        \mathbb{E} \left[ \sup_{\pi \in \Pi_k} \frac{L}{n} \sum_{i \in I_\ell} 2 \text{Rad}_i \Gamma_{1h}(T_i,X_i;\pi;g,m;\hat{g}_\ell,\hat{m}_\ell) \Bigg| \hat{g}_\ell,\hat{m}_\ell \right] \leq C_1 \frac{\lVert \hat{g}_\ell - g \rVert_\infty \lVert \hat{m}_\ell - m \rVert_\infty}{h} \sqrt{L\frac{\operatorname{VC}(\Pi_k)}{n}}.
    \end{align*}
    %Then, the desired result is achieved by taking expectations on both sides and using the Cauchy-Schwarz inequality. 
\end{proof}

\subsection{Preliminary Results in Empirical Process Theory} \label{sec:empirical-process}

For completeness of the paper, we incorporate a set of results in this section on the VC dimension and empirical process theory that are scattered in the literature. Some minor adjustments are made to suit our needs. When the statements differ from their original form in the literature, we provide the corresponding proofs. Also, in this section, we use the notation $g$ to denote a generic function instead of the inverse propensity score used in the main text.

\begin{lemma} \label{lm:tail-to-Op-rate}
Let $X_n$ be a sequence of random variables and $a_n \uparrow \infty$ be a deterministic sequence.
\begin{enumerate}
    \item If, for all $t>0$, $\mathbb{P}(X_n < t) \leq C \exp(-a_n t^2)$, then $X_n \leq O_p(a_n^{-1/2})$.
    \item If, for all $t>0$, $\mathbb{P}(X_n < t) \leq C \exp(-a_n t)$, then $X_n \leq O_p(a_n^{-1})$.
\end{enumerate}
\end{lemma}

\begin{proof}[Proof of Lemma \ref{lm:tail-to-Op-rate}]
    Denote $X_n^+ = \max\{X_n,0\}$ as the positive part of $X_n$. For any fixed $M > 0$, we have
    \begin{align*}
        \mathbb{P}(a_n^{1/2}X_n^+ > M) \leq C \exp(-a_n (M a^{-1/2})^2) = C \exp(-M^2).
    \end{align*}
    By definition, this shows that $X_n^+ = O_p(a_n^{-1/2})$. The first result follows since $X_n \leq X_n^+$. The second result can be derived analogously.
\end{proof}

\begin{lemma} \label{lm:VC-permanence}
    Let $\mathcal{G}$ be a class of functions with finite VC dimension. Let $g_1,g_2,g_3$ be fixed measurable functions. Define the following function classes:
    \begin{align*}
        g_1 \circ \mathcal{G} \equiv \{g_1 \circ g: g \in \mathcal{G}\}, \\
        g_2 \oplus \mathcal{G} \equiv \{g_2 + g: g \in \mathcal{G}\}, \\
        g_3 \otimes \mathcal{G} \equiv \{g_3 \cdot g: g \in \mathcal{G}\}.
    \end{align*}
    Their VC dimensions are bounded as 
    \begin{align*}
        \operatorname{VC}(g_1 \circ \mathcal{G}) \leq & \operatorname{VC}(\mathcal{G}), \\
        \operatorname{VC}(g_2 \oplus \mathcal{G}) \leq & \operatorname{VC}(\mathcal{G}), \\
        \operatorname{VC}(g_3 \otimes \mathcal{G}) \leq & 2\operatorname{VC}(\mathcal{G}).
    \end{align*}
    If $g_3$ is nonnegative, then $\operatorname{VC}(g_3 \otimes \mathcal{G}) \leq \operatorname{VC}(\mathcal{G})$.
\end{lemma}

\begin{proof}[Proof of Lemma \ref{lm:VC-permanence}]
    The first and second results follow from the proofs of Lemma 2.6.18(viii) and (v) in \cite{wellner1996}, respectively. The third claim can be proved based on the proof of Lemma 2.6.18(vi) in \cite{wellner1996}. For a function $g \in \mathcal{G}$, the subgraph of $g \cdot g_3$ is equal to the union of the following three disjoint sets:
    \begin{align*}
        C^+(g) & \equiv \{(x,t): t < g(x)g_3(x),g_3(x) > 0\}, \\
        C^-(g) & \equiv \{(x,t): t < g(x)g_3(x),g_3(x) < 0\}, \\
        C^0 & \equiv \{(x,t): t<0, g_3(x) = 0\}.
    \end{align*}
    Define $\mathcal{C}^+ \equiv \{C^+(g) : g \in \mathcal{G}\}$ and $\mathcal{C}^- \equiv \{C^-(g) : g \in \mathcal{G}\}$. The class of subgraphs of $\{g\cdot g_3:g \in \mathcal{F}\}$ is equal to $\mathcal{C}^+ \sqcup \mathcal{C}^- \sqcup \{C_0\}$, where $\sqcup$ denotes elementwise union. Notice that $\mathcal{C}^+$ contains subsets of the set $\{(x,t):g_3(x) > 0\}$, while $\mathcal{C}^-$ contains subsets of the set $\{(x,t):g_3(x) < 0\}$. The sets $\{(x,t):g_3(x) \leq 0\}$ and $\{(x,t):g_3(x) > 0\}$ are disjoint. Therefore, by Problem 2.6.12 in \cite{wellner1996}, $\operatorname{VC}(g_3 \otimes \mathcal{G}) = \operatorname{VC}(\mathcal{C}^+) + \operatorname{VC}(\mathcal{C}^-) + \operatorname{VC}(\{C_0\})$. Since the class $\{C_0\}$ only contains one element, its VC dimension is zero (because it cannot shatter any single point). For $\mathcal{C}^+$, suppose that a set of points $(x_i,t_i)_{i \in I}$ is shattered by $\mathcal{C}^+$. Then by defining $t_i' \equiv t_i/g_3(x_i)$, we can see that the subgraphs of $\mathcal{G}$ shatter $(x_i,t_i')_{i \in I}$. This implies that $\operatorname{VC}(\mathcal{C}^+) \leq \operatorname{VC}(\mathcal{G})$. For the class $\mathcal{C}^-$, notice that its element $C^-$ can be written as the complement of a closed subgraph:
    \begin{align*}
        C^- = \{(x,t): t/g_3(x) \leq f(x),g_3(x) < 0\}^c.
    \end{align*}
    By the proof of Lemmas 2.6.17(i) and Problem 2.6.10, we know that taking complements and changing $\leq$ to $<$ in the definition of subgraph do not alter the VC dimension; hence, $\operatorname{VC}(\mathcal{C}^-) \leq \operatorname{VC}(\mathcal{G})$. Then the result follows. In the case that $g_3$ is nonnegative, then $C^-(g)$ is empty, and hence $\operatorname{VC}(g_3 \otimes \mathcal{G}) \leq \operatorname{VC}(\mathcal{G})$.
\end{proof}

For a function class $\mathcal{G}$ of a random vector $S$, let $N(\epsilon,\mathcal{G},L_2(Q)$ be the $\epsilon$-covering number of $(\mathcal{G},L_2(Q))$, that is, the minimal number of balls (with centers in $\mathcal{G}$) of radius $\epsilon$ (under the $L_2(Q)$ metric) needed to cover $\mathcal{G}$. Denote $G$ as the envelope function of $\mathcal{G}$, that is, $|g| \leq G$ for all $g \in \mathcal{G}$.
Define the supremum of the empirical process as
\begin{align*}
    \Delta \equiv \sup_{g \in \mathcal{G}}\left| \frac{1}{n} \sum_{i=1}^n (g(S_i) - \mathbb{E}[g(S_i)]) \right|,
\end{align*}
where $\{S_i: 1 \leq i \leq n\}$ is an iid sample of $S$. Denote $\Delta_\text{Rad} $ as the corresponding Rademacher complexity after symmetrization:
\begin{align*}
    \Delta_\text{Rad} \equiv \sup_{g \in \mathcal{G}}\left| \frac{1}{n} \sum_{i=1}^n 2 \text{Rad}_i g(S_i)\right|,
\end{align*}
where $\{\text{Rad}_i:  1 \leq i \leq n \}$ is an iid sequence of Rademacher variables independent of the sample $\{S_i: 1 \leq i \leq n\}$.
Let $\sigma^2(\mathcal{G})$ and $B(\mathcal{G})$ denote the second-moment bound and uniform bound:
\begin{align*}
    \sigma^2(\mathcal{G}) & \geq \sup_{g \in \mathcal{G}} \mathbb{E}[g(S)^2], \\
        B(\mathcal{G}) & \geq \sup_{g \in \mathcal{G}} \sup_{s}|g(S)|.
\end{align*}
The following two lemmas bound the expectation $\mathbb{E}[\Delta]$ and the tail probability of the deviation $\Delta - \mathbb{E}[\Delta]$, respectively.

\begin{lemma}[Lemmas A.4 and A.5 in \cite{kitagawa2018should}]\label{lm:kitagawa} 
Assume that $B(\mathcal{G})<\infty,\operatorname{VC}(\mathcal{G})<\infty$. Suppose that there exists $V>0$, such that the covering number of $\mathcal{G}$ is bounded by
\begin{align} \label{eqn:VC-type}
    \sup_{Q}N(\epsilon \lVert G \rVert_{L_2(Q)},\mathcal{G},L_2(Q)) \leq \left( \frac{c_0}{\epsilon} \right)^{2V},
\end{align}
for a universal constant $c_0 > 0$. The supremum is taken over all probability measures on the domain of the functions in $\mathcal{G}$. 
Then there exists a universal constant (that only depends on $c_0$) such that the following two bounds hold for all $n \geq 1$.
    \begin{align}
        \mathbb{E}[\Delta_\text{Rad}] \leq & c B(\mathcal{G}) \sqrt{\frac{V}{n}}, \label{eqn:EDelta-bound1} \\
        \mathbb{E}[\Delta_{\text{Rad}}] \leq & 2c^2 B(\mathcal{G}) \frac{V}{n} + c  \sqrt{\frac{\sigma^2(\mathcal{G})V}{n}} \label{eqn:EDelta-bound2}.
    \end{align}
    In particular, due to symmetrization \citep[e.g., Theorem 2.1 in][]{koltchinskii2011oracle}, the above bounds remain valid if we replace $\Delta_{\text{Rad}}$ by $\Delta$.
\end{lemma}
\begin{proof}[Proof of Lemma \ref{lm:kitagawa}]
    The two bounds are essentially given by Lemmas A.4 and A.5 in \cite{kitagawa2018should}. The only difference is that \cite{kitagawa2018should} assume that the function class $\mathcal{G}$ has the VC dimension $V$. In our case, we relax this condition to the exponential bound on the covering number. This does not change the proof for the results because \cite{kitagawa2018should} only utilize the VC property to prove this bound on the covering number (which is presented by Lemma \ref{lm:covering-number-VC} below).
    The second bound is slightly different from the statement of Lemma A.5 in \cite{kitagawa2018should}. It follows from the last inequality in the proof of that lemma. In particular, that inequality gives
    \begin{align*}
        \mathbb{E}[\Delta] \leq & c^2 B(\mathcal{G}) \sqrt{\frac{V}{n}} \left( \sqrt{\frac{V}{n}} + \sqrt{\frac{V}{n} + \frac{\sigma^2(\mathcal{G})}{B(\mathcal{G})^2c^2}} \right) \\
        \leq & c^2 B(\mathcal{G}) \sqrt{\frac{V}{n}} \left( 2\sqrt{\frac{V}{n}} + \sqrt{\frac{\sigma^2(\mathcal{G})}{B(\mathcal{G})^2c^2}} \right) = 2c^2 B \frac{V}{n} + c  \sqrt{\frac{\sigma^2(\mathcal{G})V}{n}},
    \end{align*}
    where we have used the fact that $\sqrt{a+b} \leq \sqrt{a} + \sqrt{b}$ for any $a,b>0$.
\end{proof}

\begin{lemma}[Talagrand's inequality] \label{lm:talagrand}
    Assume that $B(\mathcal{G})<\infty$. The following right and left tail bounds hold for all $\alpha>0$.
    \begin{align}
        \mathbb{P}\left( \Delta - \mathbb{E}[\Delta] > \alpha \right) \leq & 2 \exp \left( -\frac{n\alpha^2}{8e\sigma^2(\mathcal{G} )+16eB(\mathcal{G})\mathbb{E}[\Delta] + 4B(\mathcal{G})\alpha} \right), \label{eqn:talagrand-right}\\
        \mathbb{P}\left( -(\Delta - \mathbb{E}[\Delta]) > \alpha \right) \leq & 2 \exp \left( -\frac{n\alpha^2}{2\sigma^2(\mathcal{G} )+4B(\mathcal{G})\mathbb{E}[\Delta] + 2B(\mathcal{G})\alpha} \right)\label{eqn:talagrand-left}.
    \end{align}
\end{lemma}
\begin{proof}
    Theorem 3.27 in Chapter 3 of \cite{wainwright_2019} gives the right tail bound and the bound in (3.84) following that theorem. The left tail bound is given by Theorem 1.2 in \cite{klein2005concentration}. Notice that the bounds given by \cite{klein2005concentration} are for function classes uniformly bounded by 1. After a simple rescaling, the bound can be changed to our desired bound for any uniformly bounded class.
\end{proof}

\begin{lemma}[2.6.7 Theorem in \cite{wellner1996}] \label{lm:covering-number-VC}
     For a function class $\mathcal{G}$ with finite VC dimension, we have
    \begin{align*}
        N\left(\epsilon\|G\|_{L_2(Q)}, \mathcal{G}, L_2(Q)\right) \leq \left(c_0/\epsilon\right)^{2\operatorname{VC}(\mathcal{G})}, \forall \epsilon \in (0,1),
    \end{align*}
    for a universal constant $c_0$ that can be computed explicitly.
\end{lemma}

\begin{proof}[Proof of Lemma \ref{lm:covering-number-VC}]
    By Theorem 2.6.7 in \cite{wellner1996}, the covering number is bounded as 
    \begin{align*}
        N\left(\epsilon\|G\|_{Q, 2}, \mathcal{G}, L_2(Q)\right) \leq & c_1 (\operatorname{VC}(\mathcal{G})+1)(16 e)^{\operatorname{VC}(\mathcal{G})+1}\left(1/\epsilon\right)^{2\operatorname{VC}(\mathcal{G})}, \forall \epsilon \in (0,1),
    \end{align*}
    for some universal constant $c_1$. The desired result follows by observing that the quantity 
    \begin{align*}
        \left( (\operatorname{VC}(\mathcal{G})+1)(16 e)^{\operatorname{VC}(\mathcal{G})+1}\right)^{1/(2\operatorname{VC}(\mathcal{G}))}
    \end{align*}
    is bounded for any choice $\operatorname{VC}(\mathcal{G})$.
\end{proof}

\begin{lemma}[Theorem 3 in \cite{ANDREWS1994empirical}] \label{lm:covering-number-permanence}
    Let $\mathcal{G}_1$ and $\mathcal{G}_2$ with envelopes $G_1$ and $G_2$, respectively. Define 
    \begin{align*}
        \mathcal{G}_1 \oplus \mathcal{G}_2 & \equiv \{g_1 + g_2: g_1 \in \mathcal{G}_1, g_2 \in \mathcal{G}_2\}, \\
        \mathcal{G}_1 \otimes \mathcal{G}_2 & \equiv \{g_1 \cdot g_2: g_1 \in \mathcal{G}_1, g_2 \in \mathcal{G}_2\}. 
    \end{align*}
    The classes $\mathcal{G}_1 \oplus \mathcal{G}_2$ and $\mathcal{G}_1 \otimes \mathcal{G}_2$ admit envelope $G_1 + G_2$ and $G_1 \cdot G_2$, respectively. Their covering numbers are bounded as
    \begin{align*}
        N(\epsilon \| G_1 + G_2 \|,\mathcal{G}_1 \oplus \mathcal{G}_2, L_2(Q)) \leq & N(\epsilon \| G_1 \|/2,\mathcal{G}_1 , L_2(Q)) N(\epsilon \| G_2 \|/2,\mathcal{G}_2 , L_2(Q)), \\
        \sup_{Q} N(\epsilon \| G_1 \cdot G_2 \|,\mathcal{G}_1 \otimes \mathcal{G}_2, L_2(Q)) \leq & \left( \sup_{Q} N(\epsilon \| G_1 \|/2,\mathcal{G}_1 , L_2(Q)) \right) \left(\sup_{Q} N(\epsilon \| G_2 \|/2,\mathcal{G}_2 , L_2(Q)) \right). 
    \end{align*}
    They are given by (A.5) and (A.7) in \cite{ANDREWS1994empirical}.
\end{lemma}

\end{document}